\documentclass[12pt,onesided]{article}

\usepackage[leqno]{amsmath}
\usepackage{amsthm}
\usepackage{amssymb}
\usepackage{amsfonts}
\usepackage{array}
\usepackage{hyperref}
\usepackage{eurosym}

\newcommand{\R}{\mathbb{R}}

\newcommand{\ol}{\overline}

\newcommand{\yr}[1]{\overline{#1}\vrule\ }

\newtheorem{theorem}{THEOREM} 
\newtheorem{corollary}{COROLLARY} 
\newtheorem{proposition}{PROPOSITION}
\newtheorem{lemma}{LEMMA} 
\newtheorem{definition}{DEFINITION}
 
\newtheorem{assumption}{ASSUMPTION}

\def\bi{\begin{itemize}}
\def\ei{\end{itemize}}
\def\i{\item}

\begin{document}

\title{No-Arbitrage Prices of Cash Flows and Forward Contracts as Choquet Representations}
\author{Tom Fischer\thanks{Institute for Mathematics, University of Wuerzburg, 
Emil-Fischer-Strasse 30, 97074 Wuerzburg, Germany.
Tel.: +49 931 3188911.
E-mail: {tom.fischer@uni-wuerzburg.de}.
}
\thanks{Thanks go to Torsten Kleinow for pointing out Ragnar Norberg's article on payment measures, and also
to Ragnar Norberg for a comment regarding this work.}\\
University of Wuerzburg}
\date{\today}

\maketitle

\begin{abstract}
In a market of deterministic cash flows, given as an additive, symmetric relation of exchangeability
on the finite signed Borel measures on the non-negative real time axis,
it is shown that the only arbitrage-free price functional that fulfills some
additional mild requirements is the integral of the unit zero-coupon
bond prices with respect to the payment measures. 
For probability measures, this is a Choquet representation, where the
Dirac measures, as unit zero-coupon bonds, are the extreme points. Dropping one of the requirements, 
the Lebesgue decomposition is used to construct counterexamples, where the Choquet price formula
does not hold despite of an arbitrage-free market model. The concept is then extended
to deterministic streams of assets and currencies in general, yielding a valuation principle for 
forward markets. Under mild assumptions, it is shown that a foreign cash flow's worth in
local currency is identical to the value of the cash flow in local currency for which 
the Radon-Nikodym derivative with respect to the foreign cash flow is the forward FX rate.
\end{abstract}





\section{Introduction}

\label{Introduction}

One of the first things a student of financial or actuarial mathematics commonly learns is that the 
present value
of a temporary annuity that pays one currency unit annually in advance for $n$ years is given by the partial
sum of a geometric series, $a_{\yr{n}} = v^0 + v^1 +\ldots + v^{n-1}$, 
where the common ratio $v = 1/(1+i)$ is
the discount factor belonging to the annual effective rate of interest $i>0$
(e.g.~Gerber (1997), p.~9; McCutcheon and Scott (1986), p.~45). Often, at this introductory
stage of a course in this field, it stays unmentioned that this value can in fact be understood as
a no-arbitrage price. Similarly, if one leaves
the constant interest environment and assumes effective spot rates $y_t$ for a maturity $t\geq 0$, 
the present value of the annuity
would be given by $(1+y_0)^0 + (1+y_1)^{-1} + \ldots + (1+y_{n-1})^{1-n}$.
Denote now the prices of unit zero-coupon bonds with maturity $t$ by $P_t$. Since spot rates are the yields
of zero-coupon bonds, and since such an annuity simply is a collection of unit zero-coupon bonds, 
it holds that
\begin{equation}
\label{cont_ann}
a_{\yr{n}}
= 
\sum_{t=0}^{n-1}
P_t
=
\int_0^n P_t d\gamma(t) ,
\end{equation}
where 
\begin{equation}
\gamma = \sum_{t=0}^{n-1} \delta_t 
\end{equation}
and where $\delta_t$ is the Dirac measure in $t$ (for the first equality in \eqref{cont_ann}
compare also Hull (2008), p.~79, in the case of the value of a coupon paying bond). 
The main objective of this article is to derive under
fairly mild conditions that this pricing principle in the form of an integral of the unit zero-coupon bond
prices with regards to a measure that describes the cash flow is indeed 
the uniquely determined no-arbitrage price $\pi(\gamma)$ of any $\gamma$ which is a general
deterministic cash flow in the sense that it is an arbitrary finite signed Borel measure on 
the non-negative real time axis, i.e.~it holds that
\begin{equation}
\label{price1}
\pi(\gamma) = \int_0^\infty \pi(\delta_t)\, d\gamma(t) = \int_0^\infty P_t\, d\gamma(t) .
\end{equation}
Loosely speaking, \eqref{price1} means that any payments of a cash flow get discounted with the
unit zero-coupon prices, before they get added up. However, an important point here is that while,
for given $P_t$, \eqref{price1} follows for discrete cash flows with finitely many payments immediately 
from simple no-arbitrage considerations, this is not the case for cash flows which
pay continuously over time. For instance, the present value of an annuity which time-continuously
pays 1 per year at a constant rate in a flat term structure is in the literature usually -- and without
much further comment -- given as (e.g.~McCutcheon and Scott (1986), p.~51)
\begin{equation}
\label{counterexample1}
\overline{a}_{\yr{n}} 
= 
\int_0^n P_t \, dt 
= 
\int_0^n (1+i)^{-t} \, dt .
\end{equation}
However, this is not a simple consequence of no-arbitrage pricing in relation to zero-coupon bonds, 
as an example of an arbitrage-free
market with $P_t = (1+i)^{-t}$ in Section \ref{Counterexamples} of this article will demonstrate, 
where \eqref{counterexample1} does not hold. 
So, the additional, but relatively mild, assumptions beyond no-arbitrage that are made in this paper, 
are in fact crucial for \eqref{price1} to hold in general.

While the article is mostly focused on the valuation of cash flows relative to unit zero-coupon bonds,
it is possibly worth noting that the notion of a market as a relation between cash flows will require a minimum
of assumptions, such as symmetry, additivity, and the existence of at least one strictly positive spot price 
for strictly positive cash flows. 
That the exchangeability relation is in fact an equivalence relation, or that short-selling is possible,
is a mere consequence of additivity. Similarly, linearity of the equivalence relation is derived from no-arbitrage,
and not assumed {\em a priori}.

In a second step, the article then applies the obtained pricing principle to forward markets,
meaning the consideration of a money market with deterministic cash flows in a 
local currency combined with 
a market in which deterministic flows of assets, commodities, or foreign currency are
exchanged. Mathematically, this combination is modeled by a direct sum of two linear spaces of
finite signed measures.
This results in an arbitrage-free forward contract pricing formula which states, 
for instance in the case of a FX forward market which trades the EUR/USD cross, that
(in slight abuse of later used notation)
the EUR--price 
$\pi_\text{\euro}(\gamma_\$)$ of a U.S.--Dollar cash flow $\gamma_\$$ is given by
\begin{eqnarray}
\label{price2}
\pi_\text{\euro}(\gamma_\$)
& = &
\text{USDEUR}_0 \cdot\pi_\$(\gamma_\$) 
\; = \;
\text{USDEUR}_0 \int_0^\infty P^\text{\$}_t\, d\gamma_\text{\$}(t) 
\\
\nonumber
& = & 
\int_0^\infty P^\text{\euro}_t\, d\gamma_\text{\euro}(t) 
\; = \;
\pi_\text{\euro}(\gamma_\text{\euro}) ,
\end{eqnarray}
where 
\begin{equation}
\frac{d\gamma_\text{\euro}}{d\gamma_\$}(t) = \text{USDEUR}_t ,
\end{equation}
i.e.~where $\gamma_\text{\euro}$ is the measure with the Radon-Nikodym derivative
$\text{USDEUR}_t$ with respect to $\gamma_\$$, where
$\text{USDEUR}_t$  is the forward price in EUR of USD 1.00 for delivery at time $t$,
$\text{USDEUR}_0$ is the corresponding FX spot rate,
and $P^\text{\euro}_t$ ($P^\text{\$}_t$) is the \euro-price ($\$$-price) of a zero-coupon bond that pays 
EUR 1.00 (USD 1.00) at $t$.
As is well known in theory and practice (e.g.~pp.~167--169 in Hull, 2008), Eq.~\eqref{price2} means that 
the no-arbitrage price of a FX cash flow can be calculated in two ways: 

1) discount the cash
flow in the foreign term structure and then convert it at the spot exchange rate, or 

2) 
use the prevailing forward rates, turn the entire FX cash flow into one in the local currency,
and then discount this new cash flow with respect to the local term structure.\\

While simplified versions of both valuation principles, \eqref{price1} and \eqref{price2}, 
are ubiquitously used in theory and practice, 
they are usually not formulated for such general cash flows, 
and \eqref{price1} is usually not mathematically rigorously derived as a no-arbitrage
price in a specific model of cash flows that can be exchanged for one another.
An important exception, but not an arbitrage theory approach in the strict sense,
where \eqref{price1} was formulated and derived in a somewhat different, but
also very general context, is Norberg (1990), which will be discussed in detail later 
on (Section \ref{Norberg's theory of consistent cash flow valuation}).
The presented article also makes a point in proving that examples of arbitrage-free
models of money markets or forward markets, where these price formulae hold,
indeed exist, and it explicitly shows how such models can be constructed, but also, how
arbitrage-free counterexamples, where the price formulae do not hold, can be constructed.

The outline is as follows. In the next section, several results for additive relations on
linear spaces will be derived in preparation for the application of those to market models
of deterministic cash flows. After some preliminary remarks on measure
and integration in Section \ref{Preliminary remarks on measures and integration}, the
notion of a market for cash flows will be introduced in Section \ref{Markets and prices of deterministic cash flows}
by means of an additive, symmetric relation of exchangeability. 
The results of Section \ref{Some results for additive relations on linear spaces}
imply transitivity and the possibility of short-selling.
In Section \ref{No-arbitrage prices of deterministic cash flows}, using the
results of Section \ref{Some results for additive relations on linear spaces},
the equivalence of no-arbitrage and the Law of One Price is shown, as well as the linearity of
exchangeable trades under no-arbitrage. 
Furthermore, uniqueness and linearity of the arbitrage-free price functional and a
change of numeraire formula are derived. 
The main results on no-arbitrage pricing 
relative to zero-coupon bonds by Choquet representations are given in Section
\ref{No-arbitrage pricing relative to zero-coupon bonds}. It will be seen that the additional
assumptions, that have to be made to obtain the price formula, are very natural.
The next two sections then contain examples of market models where the price formula holds, 
and counter-examples of arbitrage-free market models, where the formula does not hold because
one of the earlier mentioned assumptions is not enforced. This is followed by a consideration of
forward rates and forward prices of cash flows in Section \ref{Forward rates}.
Section \ref{Norberg's theory of consistent cash flow valuation} contains the earlier mentioned
discussion and comparison of the results in Norberg (1990). This is followed by the derivation
of valuation formulae for combined markets and/or forward FX or commodities markets in 
Section \ref{Forward markets}. A final section contains a short conclusion.


\section{Some results for additive relations on linear spaces}

\label{Some results for additive relations on linear spaces}

Consider a real vector space $\mathcal{V}$, where the null element is denoted by $o$.

\begin{assumption}
\label{assu:relation}
For the real vector space $\mathcal{V}$, there exists a relation 
$(\mathcal{V}, \sim)\subset \mathcal{V}\times\mathcal{V} $, such that
for any $\gamma_1, {\ldots},\gamma_4\in\mathcal{V}$ one has
\begin{enumerate}
\i reflexivity: $\gamma_1 \sim \gamma_1$.
\i symmetry: $\gamma_1 \sim \gamma_2$ implies $\gamma_2 \sim \gamma_1$.
\i additivity:
If $\gamma_1 \sim \gamma_2$ and $\gamma_3 \sim \gamma_4$,
then $\gamma_1 + \gamma_3 \sim \gamma_2 + \gamma_4$.
\end{enumerate}
\end{assumption}

\begin{lemma}[Invariance w.r.t.~sign]
\label{lemma:inversionc}
Under Assumption \ref{assu:relation},
if $\gamma_1,\gamma_2\in \mathcal{V}$ and $\gamma_1 \sim \gamma_2$, then
$o \sim \gamma_2 - \gamma_1 \sim \gamma_1 - \gamma_2$, and 
$- \gamma_1 \sim -\gamma_2$.
\end{lemma}

\begin{proof}
From $\gamma_1 \sim \gamma_2$ and symmetry by consecutively adding $-\gamma_1 \sim -\gamma_1$
and $-\gamma_2 \sim -\gamma_2$.
\end{proof}

\begin{lemma}[Transitivity, equivalence relation]
\label{lemma:equivalencec}
Under Assumption \ref{assu:relation},
the relation $\sim$ is an equivalence relation on $\mathcal{V}$, since 
for any $\gamma_1, \gamma_2, \gamma_3\in\mathcal{V}$, one has that
if $\gamma_1 \sim \gamma_2$ and $\gamma_2 \sim \gamma_3$,
then $\gamma_1 \sim \gamma_3$.
\end{lemma}

\begin{proof}
By Lemma \ref{lemma:inversionc}, adding
$\gamma_1 \sim \gamma_2$,  $- \gamma_2 \sim -\gamma_2$, and $\gamma_2 \sim \gamma_3$.
\end{proof}

\begin{assumption}
\label{assu:cone}
The real vector space $\mathcal{V}$ contains a positive cone
$\mathcal{V}^+\subset\mathcal{V}$ such that
\begin{enumerate}
\item
if $\alpha, \beta\in\mathcal{V}^+$ and $a, b \geq 0$, then $a\alpha+b\beta\in\mathcal{V}^+$.
\item
if $\gamma\in\mathcal{V}$, then there exist $\gamma^+, \gamma^-\in\mathcal{V}^+$
such that $\gamma =\gamma^+ - \gamma^-$.
\end{enumerate}
\end{assumption}

\begin{definition}
Under Assumption \ref{assu:relation}, let $o\neq\gamma_0\in\mathcal{V}$.
If for some $\gamma\in\mathcal{V}$ and $b\in\R$
\begin{equation}
\label{gamma sim bdelta_0c}
\gamma \sim b\gamma_0 ,
\end{equation}
then $b$ is called the $\gamma_0$-value of $\gamma$.
\end{definition}

\begin{assumption}
\label{assu:gamma-value}
Under Assumption \ref{assu:relation} and \ref{assu:cone},
there exists at least one $\gamma_0\in\mathcal{V}^+\setminus\{o\}$
such that there exists at least one strictly positive $\gamma_0$-value
for any $\gamma\in\mathcal{V}^+\setminus\{o\}$.
\end{assumption}

\begin{lemma}[Existence of a $\gamma_0$-value] 
\label{price existencec}
Under Assumption \ref{assu:relation} and \ref{assu:cone}, and if
$\gamma_0\in\mathcal{V}^+\setminus\{o\}$ is as in Assumption \ref{assu:gamma-value}, then
any $\gamma\in\mathcal{V}$ has at least one $\gamma_0$-value, and any
$\gamma \in \mathcal{V}^+\setminus\{o\}$ has at least one strictly positive $\gamma_0$-value.
\end{lemma}

\begin{proof}
The second statement holds by Assumption \ref{assu:gamma-value}.
For $\gamma\in\mathcal{V}$, one has 
$\gamma = \gamma^+ - \gamma^-$ for some $\gamma^+,\gamma^- \in \mathcal{V}^+$
by Property 2 of  Assumption \ref{assu:cone}. 
Because of the second statement, there exist 
$b^+,b^-\in\R^+_0$ such that
$b^+\gamma_0\sim \gamma^+$ and $b^-\gamma_0\sim \gamma^-$. By additivity and by
Lemma \ref{lemma:inversionc}, 
$\gamma = \gamma^+ - \gamma^-
\sim
b^+\gamma_0 - b^-\gamma_0
= (b^+-b^-)\gamma_0$. Hence, $(b^+-b^-)$ is a  $\gamma_0$-value of $\gamma$.
\end{proof}

\begin{definition}
\label{def:arbitragec}
Under Assumption \ref{assu:relation} and \ref{assu:cone},
$(\mathcal{V}, \sim)$ has the non-triviality property {\bf NT} with respect to $\mathcal{V}^+$ 
if there exists no $\gamma\in\mathcal{V}^+\setminus\{o\}$ such that
\begin{equation}
\label{eq:arbitragec}
o \sim \gamma .
\end{equation}
\end{definition}

\begin{proposition}[{\bf NT} $\Leftrightarrow$ uniqueness of $\gamma_0$-values]
\label{LOP_theoc}
Under Assumption \ref{assu:relation} and \ref{assu:cone}, and if
$\gamma_0\in\mathcal{V}^+\setminus\{o\}$ is as in Assumption \ref{assu:gamma-value}, then
{\bf NT} holds for $(\mathcal{V}, \sim)$ if and only if the $\gamma_0$-value of any $\gamma\in\mathcal{V}$
is uniquely determined. Under {\bf NT}, there therefore exists a for this specific $\gamma_0$
uniquely determined functional 
\begin{eqnarray}
\label{pi_1c}
\pi_{\gamma_0}:  \mathcal{V} & \rightarrow & \R  \index{$\pi$}
\end{eqnarray}
given by
\begin{equation}
\label{gammasimpigammadelta_0c}
\gamma \sim \pi_{\gamma_0}(\gamma)\gamma_0 .
\end{equation}
\end{proposition}

\begin{proof}
``$\Rightarrow$":
Assume that {\bf NT} holds.
By Lemma \ref{price existence} every $\gamma\in\mathcal{V}$ has a $\gamma_0$-value. Assume now
that there are two, i.e.~for $a < b$ one has $\gamma \sim a\gamma_0$ and $\gamma \sim b\gamma_0$.
Lemma \ref{lemma:inversionc} then implies $\gamma - \gamma = o \sim (b-a)\gamma_0 \in \mathcal{V}^+\setminus\{o\}$,
which violates {\bf NT}.\\
``$\Leftarrow$": Assume that {\bf NT} is violated. Then $\gamma\sim o =0\gamma_0$
for some $\gamma\in\mathcal{V}^+\setminus\{o\}$, but also $\gamma \sim b\gamma_0$
for some $b>0$ by Assumption \ref{assu:gamma-value}. Thus, $\gamma$ has the $\gamma_0$-value $0$ and
$b>0$ at the same time, and the $\gamma_0$-value is not unique. Therefore, uniqueness of the $\gamma_0$-value implies {\bf NT}.
\end{proof}

Obviously, $\pi_{\gamma_0}$ also depends on $\sim$. In that sense,
\begin{equation}
\label{pipisimc}
\pi_{\gamma_0} = \pi_{\gamma_0, \sim} .\index{$\pi_{\gamma_0}_{\sim, \gamma_0}$}
\end{equation}

\begin{definition}[{\bf NT} value w.r.t.~$\gamma_0$]
\label{pi_defc}
$\pi_{\gamma_0}(\gamma)$ of Proposition \ref{LOP_theoc} is called the {\bf NT} value of the vector 
$\gamma\in\mathcal{V}$ with respect to $\gamma_0$. 
\end{definition}

\begin{proposition}[Linearity of $\sim$ under {\bf NT}] 
\label{theo:linearityc}
Under Assumption \ref{assu:relation}, \ref{assu:cone}, and \ref{assu:gamma-value},
and if {\bf NT} holds on $(\mathcal{V}, \sim)$,
then $\gamma_1, {\ldots},\gamma_4\in\mathcal{V}$,
$\gamma_1 \sim \gamma_2$, $\gamma_3 \sim \gamma_4$, and $a,b\in\R$ implies that
$a\gamma_1 + b\gamma_3 \sim a\gamma_2 + b\gamma_4$.
\end{proposition}

\begin{proof}
Consider $\alpha, \beta\in\mathcal{V}^+$ with $\alpha \sim \beta$, and let $n$ be a natural number
larger than zero. Obviously, $\alpha/n \sim \beta/n$ holds if one of the vectors is $o$, since
{\bf NT} then implies that the corresponding other one is $o$, too.
If $\alpha/n \sim \beta/n$ does not hold for $\alpha, \beta\in\mathcal{V}^+\setminus\{o\}$, 
then Lemma \ref{price existencec} and Lemma \ref{lemma:inversionc} 
imply without loss of generality that 
$\alpha/n - \beta/n \sim b\gamma_0$ for some $b>0$
and any $\gamma_0\in\mathcal{V}^+\setminus\{o\}$ as in Assumption \ref{assu:gamma-value}. 
Adding $\beta/n \sim \beta/n$ and using
additivity yields $\alpha \sim \beta + nb\gamma_0$, to which, by Lemma \ref{lemma:inversionc}, 
$-\alpha \sim -\beta$ can be added to obtain $o \sim nb\gamma_0$.
Since $nb\gamma_0\in\mathcal{V}^+\setminus\{o\}$ by Property 1 of Assumption \ref{assu:cone},
this is a violation of {\bf NT}.
By contradiction, $\alpha/n \sim \beta/n$, and therefore, by additivity and by Lemma \ref{lemma:inversionc}, 
$q\alpha \sim q\beta$ for any rational $q\in\mathbb{Q}$.\\
Consider again $\alpha, \beta\in\mathcal{V}^+$ with $\alpha \sim \beta$, and let $r\in\mathbb{R}_0^+$. 
As earlier, $r\alpha \sim r\beta$ holds if one of the vectors is $o$, or if $r=0$.
Therefore, assume now $\alpha, \beta\in\mathcal{V}^+\setminus\{o\}$ and $r>0$.
By Proposition \ref{LOP_theoc} and Assumption \ref{assu:gamma-value}, 
$\alpha \sim \beta \sim p\gamma_0$ for some uniquely determined $p > 0$.
Furthermore, if $q_1 < r < q_2$ for $q_1, q_2\in\mathbb{Q}$, then 
$(r-q_1)\alpha \sim b_1\gamma_0$ and $(r-q_2)\alpha \sim -b_2\gamma_0$ for some $b_1, b_2 >0$
by Property 1 of Assumption \ref{assu:cone}, by Assumption \ref{assu:gamma-value}, 
and by Lemma \ref{lemma:inversionc}.
Since, by the first part, $q_1\alpha \sim q_1p\gamma_0$
and $q_2\alpha \sim q_2p\gamma_0$, one obtains $r\alpha \sim (q_1p+b_1)\gamma_0$  and
$r\alpha \sim (q_2p-b_2)\gamma_0$. 
Because of the uniqueness of the $\gamma_0$-values,
$q_1p < q_1p+b_1 = q_2p-b_2 < q_2p$.
However, since this inequality holds for any such $q_1, q_2$, it is established that
$rp = q_1p+b_1 = q_2p-b_2$ and $r\alpha \sim rp \gamma_0$.
Similarly, $r\beta \sim rp \gamma_0$, and therefore $r\alpha \sim r\beta$ holds for $r\in\mathbb{R}$.\\
Consider now $\alpha, \beta\in\mathcal{V}$ with $\alpha \sim \beta$, where one has 
the decomposition $\alpha = \alpha^+ - \alpha^-$ and
$\beta = \beta^+ - \beta^-$ for $\alpha^+,\alpha^-, \beta^+, \beta^- \in \mathcal{V}^+$
by Property 2 of Assumption \ref{assu:cone}.
Observe that $\alpha^- \sim \alpha^-$ and $\beta^- \sim \beta^-$ can be added to
\begin{equation}
\alpha^+ - \alpha^- \sim \beta^+ - \beta^-
\end{equation}
to obtain
\begin{equation}
\alpha^+ + \beta^- \sim \beta^+ + \alpha^- .
\end{equation}
For any $r\in\mathbb{R}$ it now holds that
\begin{eqnarray}
r(\alpha^+ + \beta^-) & \sim & r(\beta^+ + \alpha^-) \\
-r\alpha^- & \sim & -r\alpha^- \\
-r\beta^- & \sim & -r\beta^- ,
\end{eqnarray}
which adds up to $r\alpha \sim r\beta$. Together with additivity, this completes the proof. 
\end{proof}

Note that the trivial relation, where all vectors are equivalent to one another, is an example,
where linearity holds without {\bf NT}. So, {\bf NT} is sufficient for linearity of $\sim$,
but not necessary. 

\begin{theorem}[Linear characterization of {\bf NT}]
\label{pi_linearc}
Suppose Assumption \ref{assu:cone} holds.
Assumption \ref{assu:relation}, Assumption \ref{assu:gamma-value}, and {\bf NT}
then hold on $(\mathcal{V}, \sim)$ if and only 
if there exists a linear functional $\pi': \mathcal{V} \rightarrow \R$ with 
\begin{equation}
\pi'(\gamma)>0 \quad \text{ for all }\gamma\in\mathcal{V}^+\setminus\{o\} , 
\end{equation}
where for all $\gamma_1, \gamma_2 \in\mathcal{V}$
\begin{equation}
\label{simm_2c}
\gamma_1\sim\gamma_2 \quad \Leftrightarrow \quad \pi'(\gamma_1)=\pi'(\gamma_2) .
\end{equation}
Furthermore, if one of the two equivalent statement holds, then the property described in
Assumption \ref{assu:gamma-value} holds for any $\gamma_0\in\mathcal{V}^+\setminus\{o\}$.
Moreover, for any such $\pi'$ and any $\gamma_0\in\mathcal{V}^+\setminus\{o\}$, it holds that
\begin{equation}
\label{eq:pipigamma0}
\frac{\pi'}{\pi'(\gamma_0)}
= 
\pi_{\gamma_0} ,
\end{equation}
i.e.~\eqref{eq:pipigamma0} it is the uniquely determined {\bf NT} value with respect to $\gamma_0$.
\end{theorem}

\begin{proof}
The equivalence is proven first.\\
``$\Rightarrow$": 
From Proposition \ref{LOP_theoc} it is known that {\bf NT} implies that a unique $\gamma_0$-value functional 
$\pi_{\gamma_0}$ exists for any $\gamma_0\in\mathcal{V}^+\setminus\{o\}$ as in 
Assumption \ref{assu:gamma-value}. 
It is shown that it has the required properties.
For the proof of linearity, let $\gamma_i \sim b_i\gamma_0$ and therefore $\pi_{\gamma_0}(\gamma_i)=b_i$
for $i=1,2$. For any $a_i\in\R$ ($i=1,2$),
\begin{equation}
a_1\gamma_1 + a_2\gamma_2 \sim a_1b_1\gamma_0 + a_2b_2\gamma_0 = (a_1\pi_{\gamma_0}(\gamma_{1})+a_2\pi_{\gamma_0}(\gamma_{2}))\gamma_0
\end{equation}
by Proposition \ref{theo:linearityc}. Hence, 
$\pi_{\gamma_0}(a_1\gamma_1 + a_2\gamma_2) = a_1\pi_{\gamma_0}(\gamma_{1})+a_2\pi_{\gamma_0}(\gamma_{2})$ by \eqref{gammasimpigammadelta_0c}.
$\pi_{\gamma_0}(\gamma)>0$ for $\gamma\in\mathcal{V}^+\setminus\{o\}$
follows from Lemma \ref{price existence}.
For \eqref{simm_2c}: ``$\Leftarrow$": $\gamma_1 \sim \pi_{\gamma_0}(\gamma_1)\gamma_0 = \pi_{\gamma_0}(\gamma_2)\gamma_0 \sim \gamma_2$.
``$\Rightarrow$": If $\gamma_1\sim\gamma_2$,
then $\pi_{\gamma_0}(\gamma_1)\gamma_0 \sim \gamma_1 \sim \gamma_2 \sim \pi_{\gamma_0}(\gamma_2)\gamma_0$, and,
 by uniqueness of $\gamma_0$-values under {\bf NT}, $\pi_{\gamma_0}(\gamma_1)=\pi_{\gamma_0}(\gamma_2)$.\\
``$\Leftarrow$":
It needs to be checked that the properties of $(\mathcal{V}, \sim)$ in Assumption \ref{assu:relation} are fulfilled. 
Since $\pi'$ is linear, the Properties 1 to 3 are given through \eqref{simm_2c}.
Observe now that $\pi'/\pi'(\gamma_0)$ is a $\gamma_0$-value for any
$\gamma_0\in\mathcal{V}^+\setminus\{o\}$ since, for $\gamma\in\mathcal{V}$,  
\begin{equation}
\pi'(\gamma)
=
\frac{\pi'(\gamma)\pi'(\gamma_0)}{\pi'(\gamma_0)}
=
\frac{\pi'(\pi'(\gamma)\gamma_0)}{\pi'(\gamma_0)}
= 
\pi'\left(\frac{\pi'(\gamma)}{\pi'(\gamma_0)}\gamma_0\right)
\end{equation}
implies 
\begin{equation}
\label{for2ndclaim}
\gamma\sim(\pi'(\gamma)/\pi'(\gamma_0))\gamma_0
\end{equation}
by \eqref{simm_2c}. The property of Assumption \ref{assu:gamma-value} follows now
as $\pi'(\gamma)>0$ for $\gamma\in\mathcal{V}^+\setminus\{o\}$.
Because of this latter fact, {\bf NT} follows since $o \sim \gamma$ for some $\gamma\in\mathcal{V}^+\setminus\{o\}$
is impossible as $\pi'(o)=\pi'(0o)=0<\pi'(\gamma)$ by linearity.\\
The second claim has already been shown since \eqref{for2ndclaim} holds
for any $\gamma_0\in\mathcal{V}^+\setminus\{o\}$ under any of the two equivalent
conditions.
Since it has been shown that $\pi'/\pi'(\gamma_0)$ is a $\gamma_0$-value for any
$\gamma_0\in\mathcal{V}^+\setminus\{o\}$,
the last statement now follows from the second claim and from Proposition \ref{LOP_theoc}.
\end{proof}

It is now clear that the set of equivalence classes for any $(\mathcal{V}, \sim)$, in which {\bf NT} holds 
under Assumption \ref{assu:relation}, \ref{assu:cone}, and \ref{assu:gamma-value}, is 
for any $\gamma_0\in\mathcal{V}^+\setminus\{o\}$ the quotient space 
\begin{equation}
\mathcal{V}/\!\!\sim \; = \mathcal{V}/\pi_{\gamma_0}^{-1}(0) .
\end{equation}

\begin{corollary}[Properties of the {\bf NT} value]
\label{cor:1c}
Under Assumption \ref{assu:relation}, \ref{assu:cone}, and \ref{assu:gamma-value},
and for any $\gamma_0\in\mathcal{V}^+\setminus\{o\}$,
if {\bf NT} holds on $(\mathcal{V}, \sim)$, then
\begin{eqnarray}
\pi_{\gamma_0}(o) & = & 0 ,\\
\gamma - \pi_{\gamma_0}(\gamma)\gamma_0 & \sim & o ,\\
\label{gamma-pi(gamma)delta0c}
\pi_{\gamma_0}(\gamma - \pi_{\gamma_0}(\gamma)\gamma_0) & = & 0 .
\end{eqnarray}
\end{corollary}

\begin{proof}
By linearity of $\pi_{\gamma_0}$.
\end{proof}


\section{Preliminary remarks on measures and integration}

\label{Preliminary remarks on measures and integration}

Let $\mathcal{M}_b = \mathcal{M}_b(\R_0^+)$ denote the set of all bounded -- i.e.~finite -- signed 
Borel measures on the non-negative real numbers equipped with the
(trace) Borel $\sigma$-algebra, $(\R_0^+,\mathcal{B})$.
Denote with $\mathcal{M}_b^+=\mathcal{M}_b^+(\R_0^+)$ all non-negative measures in $\mathcal{M}_b(\R_0^+)$.
By the Hahn-Jordan decomposition, it is clear that $\mathcal{M}_b$ is the set of 
all finite linear combinations of probability measures on $(\R_0^+,\mathcal{B})$.
As such, $\mathcal{M}_b$ is a real vector space, and $\mathcal{M}_b^+\subset\mathcal{M}_b$ 
is a positive cone which fulfills Assumption \ref{assu:cone}.

Let $C_b = C_b(\R_0^+)$ denote the set of all continuous bounded functions $f: \R_0^+ \rightarrow \R$.
One can now define an integral of some $f\in C_b$ with respect to some $\mu\in\mathcal{M}_b$ by
\begin{equation}
\label{1}
\int_0^\infty f d\mu = \int_0^\infty f d\mu^+ - \int_0^\infty f d\mu^- ,
\end{equation}
where $\mu = \mu^+ - \mu^-$ is the uniquely determined Hahn-Jordan decomposition of $\mu$ into 
two finite non-negative Borel measures, $\mu^+,\mu^-\in\mathcal{M}_b^+$, such that 
$\tilde{\mu}^+ \geq\mu^+$ and $\tilde{\mu}^- \geq \mu^-$
for any $\tilde{\mu}^+,\tilde{\mu}^-\in\mathcal{M}_b^+$ with 
$\mu = \tilde{\mu}^+ - \tilde{\mu}^-$, and where integration on the
right hand side of \eqref{1} is standard Lebesgue integration with respect to (w.r.t.)
non-negative finite measures.

The Lebesgue integral on $C_b$ for non-negative bounded measures is
positively homogeneous and additive w.r.t.~the measure. 
To see that integration of functions in $C_b$ w.r.t.~bounded signed measures in $\mathcal{M}_b$ is linear
w.r.t.~these measures, i.e.~to see that
\begin{equation}
\label{2}
\int_0^\infty f d(a\mu+b\nu) = a\int_0^\infty f d\mu + b\int_0^\infty f d\nu 
\quad \text{for all } f\in C_b;\mu,\nu\in\mathcal{M}_b; a,b\in\R ,
\end{equation}
note that, if it is without loss of generality assumed that $a,b \geq 0$, then
\begin{eqnarray}
a\mu^+ + b\nu^+ & = & (a\mu+b\nu)^+ + \rho \\
a\mu^- + b\nu^- & = & (a\mu+b\nu)^- + \rho 
\end{eqnarray}
for some $\rho\in\mathcal{M}^+_b$, and by \eqref{1}, \eqref{2} follows from the positive
homogeneity and additivity for non-negative measures.

Consider now $\mu, \mu_k \in\mathcal{M}^+_b$ for $k=1,2,\ldots$. If
\begin{equation}
\mu = \sum_{k=1}^{\infty} \mu_k ,
\end{equation}
then one has for any non-negative, continuous and bounded function $f$ on $\R_0^+$
$\sigma$-additivity for integrals in the sense that
\begin{equation}
\label{4}
\int_0^\infty f d\mu = \sum_{k=1}^\infty \int_0^\infty f d\mu_k .
\end{equation}
This can easily be seen, since, for all $n=1,2,\ldots$, one obtains from \eqref{2} that
\begin{eqnarray}
0 
\leq 
\int_0^\infty f d\mu - \sum_{k=1}^n \int_0^\infty f d\mu_k 
& = &
\int_0^\infty f d\left(\mu - \sum_{k=1}^n \mu_k\right)   \\
\nonumber
& \leq &
\left(\mu(\R_0^+) - \sum_{k=1}^{n} \mu_k(\R_0^+)\right)\sup_{\R_0^+} f ,
\end{eqnarray}
where the last expression goes to zero because of the boundedness of $f$ and $\mu$.

The null measure on $(\R_0^+,\mathcal{B})$ is denoted by $o$, that is $o(B) = 0$
for all $B\in\mathcal{B}$, and $\delta_t$ is the Dirac measure in $t\geq 0$.
The trace of $\gamma\in\mathcal{M}_b(\mathbb{R}_0^+)$ on a closed, half open, or open 
interval $I\in\mathbb{R}_0^+$ is denoted by $\gamma_{I}$, which means
that for any $B\in\mathcal{B}(\mathbb{R}_0^+)$
\begin{equation}
\gamma_{I}(B)
=
\int_0^\infty \mathbf{1}_{B} d\gamma ,
\end{equation}
where $\mathbf{1}_{B}$ is the indicator function of $B$. Furthermore,
\begin{equation}
\label{eq:trace measures}
\mathcal{M}_b(I) 
:= 
\{\gamma_{I}: \gamma\in\mathcal{M}_b(\mathbb{R}_0^+)\}
\subset
\mathcal{M}_b(\mathbb{R}_0^+) .
\end{equation}

The reader is reminded that any $\mu\in\mathcal{M}_b^+$ can be expressed
by means of a uniquely determined distribution function, or measure generating function, 
$F_\mu: \R_0^+ \rightarrow \R_0^+$ where\index{$F_\mu$}
\begin{equation}
\label{muFF3}
F_\mu(t) = \mu([0,t]) \quad \text{ for all } t\geq 0 .
\end{equation}
For instance, $F_\mu(0) = \mu(\{0\})$ and
\begin{equation}
\label{muFF}
\mu((s,t]) = F_\mu(t) - F_\mu(s) .
\end{equation}
Consider now some $\mu\in\mathcal{M}_b^+([0,T])$, which, in analogy to \eqref{eq:trace measures},
is a positive finite Borel measure with support in $[0,T]$.
It then holds for the Lebesgue integral w.r.t.~$\mu$ and the Riemann-Stieltjes integral
w.r.t.~$F_\mu$ for a continuous function $f$ that
\begin{eqnarray}
\label{Stieltjes}
\int_0^\infty f(t)\, d\mu(t) 
& =  &
f(0)\mu(\{0\}) 
+
\int_{(0,T]} f(t)\, d\mu(t) \\
\nonumber
& = &
f(0)F_\mu(0) 
+
\int_0^T f(t)\, dF_\mu(t) ,
\end{eqnarray}
where $f(0)F_\mu(0)$ has to be added to the Riemann-Stieltjes integral over $[0, T]$
(to obtain the Lebesgue integral over $[0, T]$) since,  
because of \eqref{muFF}, it cannot ``see" any mass on zero.


\section{Markets of deterministic cash flows}

\label{Markets and prices of deterministic cash flows}

Consider $\mathcal{M}_b = \mathcal{M}_b(\R_0^+)$ now as the set of all deterministic cash flows
on the non-negative real time axis $\R_0^+$. For $\gamma\in\mathcal{M}_b$, $s,t\in\R_0^+$, $s<t$,
the value $\gamma([s,t])$ stands for the total amount of cash paid by $\gamma$ during the time interval $[s,t]$.
For $B\in\mathcal{B}$, $\gamma(B)$ stands for the total amount of cash paid during the (Borel-measurable)
`amount of time' $B$. 
It is assumed that positive amounts would correspond to payments to an owner of the cash
flow, while negative amounts would correspond to payments (liabilities) by the owner. 
{\em In extremis}, all liabilities could be paid to the same external claimant, but conversely, it could also be
assumed that any positive cash flow held by a market 
participant was mirrored by a negative cash flow held by another participant. However, since the model is not
concerned about the question to whom liabilities are being paid to, the notion of a `market participant' is 
essentially irrelevant.

$\mathcal{M}_b$ obviously contains discrete time cash flows, where payments are made or received at certain
points in time, as well as time-continuous cash flows, where payments are understood to be made or received
continuously at a certain rate of payment per time unit. The latter type would be described by a density 
-- i.e.~a Radon-Nikodym derivative -- with respect to the Lebesgue measure on $(\R_0^+,\mathcal{B})$.
In reality, of course, cash flows are always discrete in time. However, cash flows
which are assumed to be paid time-continuously are common in theoretical financial and
actuarial mathematics (see examples below, and a similar comment in Norberg, 1990). 
Therefore, $\mathcal{M}_b$
seems to be the natural mathematical object which generalizes the notion of a bilateral 
deterministic cash flow, when this cash flow is assumed to be free of default risk and finite in terms of
absolute amounts paid. 

In the following, markets of deterministic cash flows shall be considered.
A market is an exchange. The crucial question in a market at present time 0 is therefore: 
What can be exchanged for what?
This is answered by the following definition. 

\begin{definition}[Market, price]
\label{market}
A market of deterministic cash flows is given by $(\mathcal{M}_b, \sim)$, where
$\sim$ is a relation on $\mathcal{M}_b$ such that 
for any $\gamma_1, {\ldots},\gamma_4\in\mathcal{M}_b$, one has
\begin{enumerate}
\i reflexivity: $\gamma_1 \sim \gamma_1$.
\i symmetry: $\gamma_1 \sim \gamma_2$ implies $\gamma_2 \sim \gamma_1$.
\i additivity:
If $\gamma_1 \sim \gamma_2$ and $\gamma_3 \sim \gamma_4$,
then $\gamma_1 + \gamma_3 \sim \gamma_2 + \gamma_4$.
\end{enumerate}
If for some $\gamma\in\mathcal{M}_b$ and $b\in\R$
\begin{equation}
\label{gamma sim bdelta_0}
\gamma \sim b\delta_0 ,
\end{equation}
then $b$ is called a price, or -- more precisely -- a spot price of $\gamma$ in this market. 
It must now also hold for $(\mathcal{M}_b, \sim)$ that
\begin{itemize}
\item[4.]
for any $\gamma\in\mathcal{M}_b^+\setminus\{o\}$ there exists at least 
one strictly positive price.
\end{itemize}
\end{definition}

It is thus denoted by $\gamma_1\sim\gamma_2$, if $\gamma_1$ can be exchanged, or traded
for, $\gamma_2$ in the market at time 0 (`at spot'), where it is assumed that this trade can 
occur at any natural quantity. The properties 1 to 3 are therefore entirely natural,
since, first, a cash flow can be exchanged for itself, since, second, 
for any party in a trade there must be a counterparty,
which then experiences the mirrored exchange,
and since, third, it is reasonable to assume that two different trades can be executed simultaneously.
This last feature, and why no divisibility of trades was assumed, will be discussed in more detail in the
next section. By the symmetry property, it is clear that the interest for borrowing and lending
in such a model is identical. In a market of peers, for instance, in an interbank market of similarly
suited banks, this would be a realistic assumption.

The reason for the definition of `price' is clear: The cash flow $\gamma$
is equivalent to, or exchangeable for, the payment of $b$ currency units at time $0$.
However, at this point it should be mentioned that the entire rest of this article could
be adapted to the situation where a different `price' was chosen by means of an
exchange rate w.r.t.~any other fixed $\gamma_0\in\mathcal{M}_b^+\setminus\{o\}$
such that (for a different $b$)
\begin{equation}
\gamma 
\sim
b\gamma_0 .
\end{equation}
This means that $\gamma_0$ would play the role of the numeraire. 
While this would be possible, the natural choice for the numeraire is, of course, $\delta_0$.

Since trades do not always involve cash prices -- as the simple example of 
swap agreements, where typically no cash changes hands at time zero, demonstrates --
it would seem appropriate to define a market for cash flows without assuming
{\em a priori} given prices, since prices finally are only a consequence of exchangeability.
However, the definition of a market would not seem very realistic if it would
potentially exclude large sets of cash flows from having a spot price, which is why Property 4
was included in Definition \ref{market}.
The price assumption (Property 4)  is natural in the sense that a strictly positive cash flow should have some
value at present. This gives the market a minimum in structure beyond the mere possibility of the
combination of trades and allows for inter-temporal substitution.

The following lemma shows that any trade can be inverted with regards to the sign of the cash flows.
Observe that this `short-selling' of cash flows does not need to be assumed, but follows from additivity.
For the lemma, and for the rest of this paper, it is important that, by the Hahn-Jordan decomposition,
the positive cone $\mathcal{M}_b^+\subset\mathcal{M}_b$ fulfills Assumption \ref{assu:cone}, and that a
market $(\mathcal{M}_b, \sim)$ as in Definition \ref{market} fulfills Assumption \ref{assu:relation}
and Assumption \ref{assu:gamma-value}.

\begin{lemma}[Short-selling]
\label{lemma:inversion}
In a market $(\mathcal{M}_b, \sim)$,
if $\gamma_1,\gamma_2\in \mathcal{M}_b$ and $\gamma_1 \sim \gamma_2$, then
$o \sim \gamma_2 - \gamma_1 \sim \gamma_1 - \gamma_2$ and 
$- \gamma_1 \sim -\gamma_2$.
\end{lemma}

\begin{proof}
Lemma \ref{lemma:inversionc}.
\end{proof}

Since it is not unrealistic to assume that several trades can be carried out consecutively in, essentially, the same moment, 
the next lemma's statement is a reasonable addition to the model. Moreover, it shows that the set of all cash flows for
which a particular cash flow can be exchanged for, i.e.~the set of immediately attainable cash flows when holding this
cash flow, is an equivalence class under $\sim$.

\begin{lemma}[Transitivity, equivalence relation]
\label{lemma:equivalence}
For a market $(\mathcal{M}_b, \sim)$, $\sim$ is an equivalence relation on $\mathcal{M}_b$, since 
for any $\gamma_1, \gamma_2, \gamma_3\in\mathcal{M}_b$, one has that
if $\gamma_1 \sim \gamma_2$ and $\gamma_2 \sim \gamma_3$,
then $\gamma_1 \sim \gamma_3$.
\end{lemma}

\begin{proof}
Lemma \ref{lemma:equivalencec}.
\end{proof}

\begin{proposition}[Existence of a price] 
\label{price existence}
In a market $(\mathcal{M}_b, \sim)$, any $\gamma\in\mathcal{M}_b$ has at least one spot price, and any
$\gamma \in \mathcal{M}_b^+\setminus\{o\}$ has at least one strictly positive spot price.
\end{proposition}

\begin{proof}
Lemma \ref{price existencec}.
\end{proof}

The proposition bears some resemblance to the completeness of stochastic market
models, meaning market models with time-dynamic trading of usually finitely 
many underlying stochastic assets, 
where it is demanded, that any payoff function can be attained by a self-financing strategy of some
initial price. This price must not be unique if arbitrage is permitted. 
As a consequence of Property 4 of Def.~\ref{market}, it is now known in the here presented case
that any cash flow in $\gamma\in\mathcal{M}_b$ can be purchased for a certain spot
price, while this price is not necessarily unique.


\section{No-arbitrage prices of deterministic cash flows}

\label{No-arbitrage prices of deterministic cash flows}

\begin{definition}[Arbitrage]
\label{def:arbitrage}
There exists an arbitrage opportunity in the market $(\mathcal{M}_b, \sim)$ if
\begin{equation}
\label{eq:arbitrage}
o \sim \gamma \quad \text{ for some }\gamma\in\mathcal{M}_b^+\setminus\{o\} .
\end{equation}
\end{definition}

Expression \eqref{eq:arbitrage} means that a market participant can trade
the null cash flow, i.e.~`nothing', for a cash flow $\gamma$ that will pay a positive
amount of money over time. Such a `free lunch', however, should not be possible in an
efficient market. If there are no arbitrage opportunities, the market can be described as arbitrage-free, 
or, equivalently, it can be said that no-arbitrage, or {\bf NA}, holds. Obviously, {\bf NA}
is the analog of {\bf NT} in Definition \ref{def:arbitragec}.

\begin{definition}[Law of One Price]
\label{LOP_def}
The Law of One Price $(${\bf LOP}$)$ holds in the market $(\mathcal{M}_b, \sim)$
if for any $\gamma\in\mathcal{M}_b$ there exists exactly one cash price.
\end{definition}

\begin{theorem}[{\bf NA} $\Leftrightarrow$ {\bf LOP}]
\label{LOP_theo}
In a market $(\mathcal{M}_b, \sim)$, {\bf NA} holds if and only if {\bf LOP} holds.
In an arbitrage-free market, there therefore exists a uniquely determined functional 
\begin{eqnarray}
\label{pi_1}
\pi:  \mathcal{M}_b & \rightarrow & \R  \index{$\pi$}
\end{eqnarray}
given by
\begin{equation}
\label{gammasimpigammadelta_0}
\gamma \sim \pi(\gamma)\delta_0 .
\end{equation}
\end{theorem}

\begin{proof}
Proposition \ref{LOP_theoc}, where $\gamma_0 = \delta_0\in\mathcal{M}^+_b$.
\end{proof}

Note that in stochastic models the Law of One Price generally is not equivalent to {\bf NA},
but weaker. This was pointed out by Courtault et al.~(2004),
where, in some analogy to the stochastic Fundamental Theorem of Asset Pricing, 
{\bf LOP} is connected to the existence of certain martingales. However, even in these
models, {\bf LOP} implies {\bf NA} if the market is complete. With the earlier pointed out
similarity of the here presented model to complete stochastic models, Theo.~\ref{LOP_theo}
seems to fit in with this observation.

\begin{definition}[Present value, {\bf NA} price]
\label{pi_def}
In an arbitrage-free market $(\mathcal{M}_b, \sim)$,
$\pi(\gamma)$ of Theorem \ref{LOP_theo} is called the {\bf NA} price, the fair value, or the present value
of the cash flow $\gamma\in\mathcal{M}_b$. 
\end{definition}

$\pi$ depends on $\sim$, thus
\begin{equation}
\label{pipisim}
\pi = \pi_{\sim} .\index{$\pi_{\sim}$}
\end{equation}

\begin{theorem}[Linearity of exchangeability under {\bf NA}] 
\label{theo:linearity}
If $\gamma_1, {\ldots},\gamma_4\in\mathcal{M}_b$,
$\gamma_1 \sim \gamma_2$, $\gamma_3 \sim \gamma_4$, and $a,b\in\R$,
then $a\gamma_1 + b\gamma_3 \sim a\gamma_2 + b\gamma_4$ holds if
the market $(\mathcal{M}_b, \sim)$ is free of arbitrage.
\end{theorem}

\begin{proof}
Proposition \ref{theo:linearityc}.
\end{proof}

Analog to the remark below Proposition \ref{theo:linearityc}, the trivial market, where all cash flows can be traded for
one another, is an example, where linearity holds without {\bf NA}, meaning that absence of arbitrage is sufficient for
linearity of $\sim$, but not necessary. 

It is possibly remarkable to a certain degree, that linearity of exchangeability
was not demanded by definition, but only the much weaker and more natural additivity of trades,
which in combination with existence of prices for non-negative cash flows and {\bf NA} then implied
linearity.
For instance, also the assumption of arbitrary divisibility or scalability of trades would have been much less natural
than additivity. If $\sim$ means exchangeability at any natural quantity, it is easy to imagine a
market participant who would execute two trades simultaneously in the following sense. Firstly,
she would simply perceive the two received cash flows as a unit (e.g.~being paid into or being paid from
the same account). Secondly, the (one) cash flow that she traded in, she could indeed trade in since,
even if it was technically impossible to simply split the cash flow in two, she herself would be able to do
so as an intermediary -- for the positive part of that cash flow -- 
by passing on received payments to the new owners (by dividing them up accordingly), or 
-- for the negative part of that cash flow -- 
by passing on received payments from the new owners (in sum) to the claimant of these payments.
A similar argument would not work if, say, a trade was tried to be carried out at half size, because another
party would have to be found that would want to execute the exact same trade such that those trades
could be bundled to create the trade that the market actually permits. 
However, nothing (in the definitions) would guarantee the existence of such another party 
with those exact same wishes, even if there was an intermediary willing to bundle the trades.

Prices also are a convenient way of determining which goods\index{good} or assets can be
exchanged for one another. The following theorem clarifies this in the context of this model,
where -- not surprisingly -- the equivalence classes of exchangeable cash flows are those of
the same price.

\begin{theorem}[Linear characterization of {\bf NA}]
\label{theo:pi_linear}
\label{Theo4}
$(\mathcal{M}_b, \sim)$ is a market where {\bf NA} holds if and only if there exists a
linear functional $\pi': \mathcal{M}_b \rightarrow \R$ with 
\begin{equation}
\pi'(\gamma)>0 \quad \text{ for all }\gamma\in\mathcal{M}_b^+\setminus\{o\} , 
\end{equation}
where for all $\gamma_1, \gamma_2 \in\mathcal{M}_b$
\begin{equation}
\label{simm_2}
\gamma_1\sim\gamma_2 \quad \Leftrightarrow \quad \pi'(\gamma_1)=\pi'(\gamma_2) .
\end{equation}
The uniquely determined {\bf NA} price is given by
\begin{equation}
\label{eq:pipigamma0b}
\pi
=
\frac{\pi'}{\pi'(\delta_0)} .
\end{equation}
\end{theorem}

\begin{proof}
Theorem \ref{pi_linearc}.
\end{proof}

As similarly remarked in Section \ref{Some results for additive relations on linear spaces} and
in analogy to the stochastic theory of complete markets, the set of
equivalence classes of exchangeable cash flows in an arbitrage-free market is the quotient space 
\begin{equation}
\mathcal{M}_b/\sim \; = \mathcal{M}_b/\pi^{-1}(0) .
\end{equation}
Linearity of the {\bf NA} price followed immediately from the linearity
of the market's equivalence relation under {\bf NA} (Theorem \ref{theo:linearity}). 
Often, linearity of a price is derived by arguing that a cash flow could otherwise be
broken up into parts and be sold/bought individually to create arbitrage. This argument 
requires, of course, that a linear combination of cash flows can instantly be sold/bought for
(is equivalent to) the corresponding linear combination of those cash flows' individual prices 
(the corresponding linear combination of Dirac measures in $t=0$). Linearity, and
in particular the not so natural divisibility (see earlier comment) of exchangeability
are therefore quietly assumed, while in this article, those were derived from more
basic, but explicitly stated properties of the market (e.g.~additivity), assuming absence of arbitrage.

For the following, it will be expressed as $\gamma_1 > \gamma_2$ if 
$\gamma_1, \gamma_2\in\mathcal{M}_b$ with $\gamma_1(B) \geq \gamma_2(B)$
for all $B\in\mathcal{B}$ and $\gamma_1(B) > \gamma_2(B)$ for at least
one $B\in\mathcal{B}$ or, equivalently, if $\gamma_1 - \gamma_2 \in\mathcal{M}_b^+\setminus\{o\}$.

\begin{corollary}[Properties of the {\bf NA} price]
\label{cor:1}
In an arbitrage-free market $(\mathcal{M}_b, \sim)$, 
\begin{eqnarray}
\pi(o) & = & 0 ,\\
\gamma - \pi(\gamma)\delta_0 & \sim & o ,\\
\label{gamma-pi(gamma)delta0}
\pi(\gamma - \pi(\gamma)\delta_0) & = & 0 .
\end{eqnarray}
Furthermore,  $\pi(\gamma_1) > \pi(\gamma_2)$ for any $\gamma_1, \gamma_2\in\mathcal{M}_b$ with
$\gamma_1 > \gamma_2$.
\end{corollary}

\begin{proof}
The first three statements from Cor.~\ref{cor:1c}. The last one follows from Property 4 of Def.~\ref{market}, 
price uniqueness, and additivity, since
$\pi(\gamma_1 - \gamma_2) > 0$ for $\gamma_1 - \gamma_2 \in\mathcal{M}_b^+\setminus\{o\}$.  
\end{proof}

Eq.~\eqref{gamma-pi(gamma)delta0} means that a cash flow $\gamma$ together with the immediate payment
of its price, $\pi(\gamma)\delta_0$, has price 0 under {\bf NA}.

\begin{corollary}[Change of numeraire]
\label{cor:Change of numeraire}
In an arbitrage-free market $(\mathcal{M}_b, \sim)$ with the {\bf NA} spot price $\pi$,
the {\bf NA} price $\pi_{\gamma_0}$ w.r.t.~any other numeraire $\gamma_0\in\mathcal{M}_b^+\setminus\{o\}$
is given by
\begin{equation}
\label{eq:Change of numeraire}
\pi_{\gamma_0}
=
\frac{\pi}{\pi(\gamma_0)} .
\end{equation}
\end{corollary}

\begin{proof}
Theorem \ref{pi_linearc}.
\end{proof}


\section{No-arbitrage pricing relative to zero-coupon bonds}

\label{No-arbitrage pricing relative to zero-coupon bonds}

\begin{definition}[UZCB prices]
\label{P_t_def}
In an arbitrage-free market $(\mathcal{M}_b, \sim)$,
the by Def.~\ref{pi_def} and by Theo.~\ref{theo:pi_linear} uniquely defined 
strictly positive function
\begin{eqnarray}
\label{P_t}\index{$P_t$}
P_\cdot: \R^+_0 & \rightarrow & \R^+ \\
\nonumber
t & \mapsto & P_t = \pi(\delta_t)
\end{eqnarray}
maps any maturity $t$ to the (spot) price of the unit zero-coupon bond (UZCB) $\delta_t$.
\end{definition}

Similar to earlier, it is important to note that $P_t$ depends on $\sim$. 
In that sense
\begin{equation}
\label{pipisim_2}
P_t = P_{\sim,t} .
\end{equation}
Clearly, $P_0 = 1$. Formula \eqref{price1} can now be established for finite discrete time
cash flows.

\begin{proposition}[Choquet formula for finite discrete cash flows]
\label{Prop2}
In an arbitrage-free market $(\mathcal{M}_b, \sim)$, the price of a finite discrete time cash flow
$\gamma\in\mathcal{M}_b$
with $k=1,{\ldots},n$ payments $c_{t_k}\in\R$ at times $t_k\in\R^+_0$, i.e.
$\gamma = \sum_{k=1}^nc_{t_k}\delta_{t_k}$, is given by
\begin{equation}
\label{pi_3}
\pi(\gamma) 
= 
\sum_{k=1}^n c_{t_k}\pi(\delta_{t_k}) 
=
\sum_{k=1}^n c_{t_k}P_{t_k} 
= 
\int_0^\infty P_t\, d\gamma(t) .
\end{equation}
\end{proposition}

\begin{proof}
Def.~\ref{P_t_def} and Theo.~\ref{theo:pi_linear}.
\end{proof}

Note that \eqref{pi_3} is obviously the price that market participants incur when they replicate
the cash flow $\gamma$ using unit zero-coupon bonds. In the following,
the focus lies on the question whether the pricing formula
$\pi(\gamma) = \int_0^\infty P_t\, d\gamma(t)$ can also be applied to
more general cash flows in $\mathcal{M}_b$.
However, the assumptions so far allow no direct derivation, as this seems 
to demand stronger requirements on the price functional $\pi$. For instance,
it would seem rational to suppose
\begin{equation}
P_s \stackrel{s \rightarrow t}{\longrightarrow} P_t ,
\end{equation}
since, for a market participant, payments of identical size that are only an arbitrarily 
small time interval apart should essentially be indistinguishable. Hence, $P_\cdot: t \mapsto P_t$ 
is assumed to be a continuous real function on $\R^+_0$. Experience also shows that it
is reasonable to suppose boundedness.

\begin{assumption}[Continuous, bounded UZCB prices]
\label{P_t continuous}
\begin{equation}
P_\cdot \in C^{++}_b(\R_0^+) ,
\end{equation}
where $C^{++}_b(\R_0^+)$ denotes the strictly positive, bounded, continuous real functions
on $\R_0^+$.
\end{assumption}

Note that this assumption also means that the term structure of the yields of
zero-coupon bonds, that is $y_t = (1/P_t)^{1/t}-1$ for $t\geq 0$ (a.k.a.~the
{\em term structure of interest rates}), is continuous, but not necessarily bounded.

In this context it should be remarked that some central banks, 
for instance the German {\em Bundesbank} (e.g.~Deutsche Bundesbank, 1997),
but also the Swiss National Bank, offer yields of hypothetical zero-coupon bonds in the form 
of a parametrized smooth function (by means of the Svensson method in the
case of the Bundesbank; cf.~Schich, 1997), which, in turn, means that, even for practical purposes,
$P_\cdot$ can sometimes assumed to be given as a continuous function for maturities between, 
in the case of the Bundesbank, 3 months and up to 30 years. 
Assumption \ref{P_t continuous} is therefore fairly realistic.

\begin{assumption}[Average value requirement]
\label{mean value}
Given a market $(\mathcal{M}_b, \sim)$ for which {\bf NA} and Assumption \ref{P_t continuous}
apply, it holds for any $\gamma\in\mathcal{M}^+_b$ and for $0\leq s<t$ that
there exists some $b$ with
\begin{equation}
\min\{P_r: r\in [s,t]\} \; \leq \; b \; \leq \; \max\{P_r: r\in [s,t]\} 
\end{equation}
such that
\begin{equation}
\label{gammastsim}
\pi(\gamma_{(s,t]}) = b\gamma((s,t]) .
\end{equation}
\end{assumption}
This assumption requires for a non-negative cash flow on $(s,t]$ that its price relative
to its size can be expressed as a weighted average of the unit zero-coupon bond prices
on $[s,t]$. This is reasonable to assume, as it would make no financial sense if an
amount $\gamma((s,t])$ paid between $s$ and $t$ had a higher price than the amount times the highest
price of any unit zero-coupon bond that matures during this time interval.
The reason is that the times with the highest UZCB prices would indicate which maturities the market prefers
over others, so spreading the payment should lead to no increase in price. Similarly for the lower boundary.

It is now possible to extend the price formula to $\mathcal{M}_b([0,T])$, $T>0$, 
the bounded signed Borel measures on $[0, T]$ in the sense of \eqref{eq:trace measures}.

\begin{proposition}[Choquet formula for time-bounded cash flows]
\label{Prop3}
Given an arbitrage-free market $(\mathcal{M}_b(\R_0^+), \sim)$ for which 
Assumption \ref{P_t continuous} and \ref{mean value} apply, 
it holds for the uniquely determined no-arbitrage price $\pi(\gamma)$ of any
$\gamma\in\mathcal{M}_b([0,T])$ in this market that
\begin{equation}
\pi(\gamma) = \int_0^\infty P_t\, d\gamma(t) .
\end{equation}
\end{proposition}

\begin{proof}
First, let $\gamma\in\mathcal{M}^+_b([0,T])$.
Consider a finite partition $\mathcal{P}$ of $[0,T]$ which is given by
\begin{equation}
[0,T] = \{t_0\} \cup \bigcup_{i=0}^{n-1} (t_i,t_{i+1}]
\end{equation}
where $0=t_0<t_1<{\ldots}<t_n=T$. Now,
\begin{equation}
\gamma 
= 
\gamma(\{0\})\delta_0 + \sum_{i=0}^{n-1} \gamma_{(t_i,t_{i+1}]} .
\end{equation}
From Assumption \ref{mean value},
\begin{equation}
\pi(\gamma_{(t_i,t_{i+1}]}) = b_i\gamma((t_i,t_{i+1}])
\end{equation}
for some $b_i$ with
\begin{equation}
\label{LowUp}
\min\{P_r: r\in [t_i,t_{i+1}]\} \; \leq \; b_i \; \leq \; \max\{P_r: r\in [t_i,t_{i+1}]\} .
\end{equation}
By Theo.~\ref{theo:pi_linear} and Eq.~\eqref{muFF},
\begin{equation}
\label{LowUp2}
\pi(\gamma) 
= 
P_0\gamma(\{0\}) + \sum_{i=0}^{n-1} b_i\gamma((t_i,t_{i+1}])
=
P_0F_\gamma(0) + \sum_{i=0}^{n-1} b_i\left[F_\gamma(t_{i+1}) - F_\gamma(t_i)\right] .
\end{equation}
Because \eqref{LowUp} and \eqref{LowUp2} hold for any finite partition $\mathcal{P}$ of $[0,T]$,
this means that the second summand ($\sum_{i=0}^{n-1} \ldots$) in \eqref{LowUp2} lies between 
the upper and the lower Darboux sum for any such partition. By \eqref{Stieltjes}, and since
$P_t$ is a continuous function by Assumption \ref{P_t continuous},
\begin{equation}
\pi(\gamma)
= 
P_0F_\gamma(0) 
+
\int_0^T P_t\, dF_\gamma(t)
= 
\int_0^\infty P_t\, d\gamma(t) .
\end{equation}
Considering the Hahn-Jordan decomposition of any $\gamma\in\mathcal{M}_b([0,T])$,
the assertion now follows by linearity of $\pi$ together with \eqref{1}.
\end{proof}

A last assumption needs to be made for the final result. 

\begin{assumption}[$\sigma$-additivity of $\sim$ on $\mathcal{M}^+_b$]
\label{sigma-additivity}
For $k=1,2,\ldots$ let $\beta_k, \gamma_k\in\mathcal{M}_b^+$, let
$\beta_k \sim \gamma_k$, and let
$\sum_{k=1}^\infty \beta_k, \sum_{k=1}^\infty \gamma_k\in\mathcal{M}_b^+$. Then,
\begin{equation}
\sum_{k=1}^{\infty} \gamma_k 
\sim 
\sum_{k=1}^\infty \beta_k .
\end{equation}
\end{assumption}

\begin{lemma}[$\sigma$-additivity of $\pi$ on $\mathcal{M}^+_b$]
\label{lemma:sigma-additivity}
Under Assumption \ref{sigma-additivity}, let $\pi$ be a {\bf NA} price in a market $(\mathcal{M}_b, \sim)$.
For $k=1,2,\ldots$ let $\gamma_k\in\mathcal{M}_b^+$ and 
$\sum_{k=1}^{\infty} \gamma_k\in\mathcal{M}_b^+$. Then,
\begin{equation}
\pi\left(\sum_{k=1}^{\infty} \gamma_k\right)
=
\sum_{k=1}^{\infty} \pi(\gamma_k) .
\end{equation}
\end{lemma}

\begin{proof}
Because of Proposition \ref{price existence}, additivity, and Corollary \ref{cor:1},
it holds for all natural $n \geq 1$ that
\begin{equation}
0 
\leq 
\sum_{k=1}^{n} \pi(\gamma_k)
=
\pi\left(\sum_{k=1}^{n} \gamma_k\right)
\leq
\pi\left(\sum_{k=1}^{\infty} \gamma_k\right) 
<
\infty ,
\end{equation}
which implies $\sum_{k=1}^{\infty} \pi(\gamma_k) < \infty$,
and thus $\sum_{k=1}^{\infty} \pi(\gamma_k)\delta_0\in\mathcal{M}_b^+$.
Since $\gamma_k\sim\pi(\gamma_k)\delta_0$ for all $k$,
Assumption \ref{sigma-additivity} and uniqueness of the price prove the claim by
\begin{equation}
\sum_{k=1}^{\infty} \gamma_k
\sim
\sum_{k=1}^{\infty} \pi(\gamma_k)\delta_0 .
\qedhere
\end{equation}
\end{proof}

By Lemma \ref{lemma:sigma-additivity}, Assumption \ref{sigma-additivity} implies that a pricing principle
must also avoid large number arbitrage, i.e.~it cannot be that the price of a countable
sum of cash flows is different from the countable sum of the prices
of the individual cash flows. For instance, selling a cash flow and hedging it (in the limit) by
immediately, but individually, buying more and more fractions of it, will not create an arbitrage 
in the limit. 
Similar as in the case of the Law of One Price, it is somewhat difficult to draw
parallels to the stochastic theory of large number arbitrage,
where sequences of financial markets with an increasing number of tradeable assets are considered
(e.g.~Kabanov and Kramkov, 1998), and where, quite naturally, it is as well of interest to examine
the consequences of the exclusion of such arbitrage opportunities (e.g.~Klein, 2000). 

\begin{theorem}[Choquet formula for general cash flows]
\label{Theo3}
Given an arbitrage-free market $(\mathcal{M}_b, \sim)$ for which 
Assumption \ref{P_t continuous}, \ref{mean value}, and \ref{sigma-additivity}
apply, it holds for the uniquely determined no-arbitrage price $\pi(\gamma)$ 
of any $\gamma\in\mathcal{M}_b$ that
\begin{equation}
\label{main}
\pi(\gamma) = \int_0^\infty \pi(\delta_t)\, d\gamma(t) = \int_0^\infty P_t\, d\gamma(t) .
\end{equation}
\end{theorem}

\begin{proof}
First, let $\gamma\in\mathcal{M}^+_b$, and observe that
\begin{equation}
\gamma 
= 
\gamma(\{0\})\delta_0 + \sum_{k=1}^\infty \gamma_{(k-1,k]} ,
\end{equation}
where
$$
\pi(\gamma_{(k-1,k]}) 
=
\int_0^\infty P_t\, d\gamma_{(k-1,k]}(t) \quad \text{for } k=1,2,\ldots
$$
by Proposition \ref{Prop3}, since $\gamma_{(k-1,k]}\in\mathcal{M}^+_b([0,k])$.
Lemma \ref{lemma:sigma-additivity} and \eqref{4} now imply
\begin{eqnarray}
\pi(\gamma)
& = &
\pi(\gamma(\{0\})\delta_0) 
+
\sum_{k=1}^\infty \pi(\gamma_{(k-1,k]}) \\
\nonumber
& = &
P_0\gamma(\{0\}) 
+
\sum_{k=1}^\infty \int_0^\infty P_t\, d\gamma_{(k-1,k]}(t) 
\; = \;
\int_0^\infty P_t\, d\gamma(t) .
\end{eqnarray}
The result for signed measures again follows from the
consideration of the Hahn-Jordan decomposition of any $\gamma\in\mathcal{M}_b$,
together with the linearity of $\pi$ and \eqref{1}.
\end{proof}

Choquet's Theorem (for a textbook reference see Phelps, 2001) 
states that, for a compact convex subset $\mathcal{S}$ of a normed 
vector space $\mathcal{V}$ and for any $s\in \mathcal{S}$, there exists a probability measure $\mu_s$ on
the extreme points $\mathcal{E}\subset\mathcal{S}$ such that for any linear functional $\pi$ on $\mathcal{S}$
\begin{equation}
\label{6}
\pi(s) = \int_{\mathcal{E}} \pi(e) d\mu_s(e) .
\end{equation}
A generalization to not necessarily metrizable spaces exists (Bishop and de Leeuw, 1959).
The relation to Theorem \ref{Theo3} becomes clearer if $\pi$ is considered
on $M^+_1(\R_0^+)$, the set of probability measures on $(\R_0^+,\mathcal{B})$,
which is convex, and if it is recalled that the Dirac measures $\mathcal{D}(\R_0^+)$ 
are the extreme points of $\mathcal{M}^+_1(\R_0^+)$. One therefore has in \eqref{main} the special
case of \eqref{6} where $s  = \mu_s = \gamma$, and by identification of $r$ and $\delta_r$, the set of
extreme points $\mathcal{E}=\mathcal{D}(\R_0^+)$ is equipped with the Borel $\sigma$-algebra
on $\mathbb{R}_0^+$. While an obvious norm on $\mathcal{M}_b$ would be given by
\begin{equation}
||\gamma|| 
= 
\gamma^+(\mathbb{R}_0^+) + \gamma^-(\mathbb{R}_0^+) ,
\end{equation}
$\mathcal{M}^+_1(\R_0^+)$ is not compact in the topology belonging to this norm since
$||\delta_t - \delta_s|| = 2$ for all $s\neq t$, $s,t\geq 0$, and therefore no sequence of Dirac measures
contains an accumulation point. However, although the requirements for Choquet's Theorem are
not fulfilled here, one has for any $\gamma\in\mathcal{M}_b$ that
\begin{equation}
\label{Choquet representation}
\gamma
= 
\int_{0}^\infty \delta_t d\gamma
\end{equation}
since, for any Borel set $B$,
\begin{equation}
\gamma(B) 
=
\int_B d\gamma
=
\int_{0}^\infty \mathbf{1}_B(t) d\gamma(t)
= 
\int_{0}^\infty \delta_t(B) d\gamma(t) .
\end{equation}
So, quite naturally, any such $\gamma$ is its own Choquet representation with respect to the Dirac measures,
and a formula of type \eqref{main} for a linear price functional should not surprise, even if counterexamples 
exist (see Sec.~\ref{Counterexamples}).


\section{Examples}

\begin{proposition}[Existence of {\bf NA}-markets]
\label{existence}
Let $f\in C^{++}_b(\R_0^+)$ be such that $f(0)=1$. Define for any $\gamma\in\mathcal{M}_b$
\begin{equation}
\label{41}
\pi'(\gamma) = \int_0^\infty f(t)\, d\gamma(t) ,
\end{equation}
and define the relation $\sim$ on $\mathcal{M}_b$ by
\begin{equation}
\label{42}
\gamma_1 \sim \gamma_2 
\quad \Leftrightarrow \quad
\pi'(\gamma_1) = \pi'(\gamma_2) 
\end{equation}
for any $\gamma_1, \gamma_2\in\mathcal{M}_b$.
Then, $(\mathcal{M}_b, \sim)$ is a market for deterministic cash flows as in Def.~\ref{market}
which is free of arbitrage, and the Assumptions \ref{P_t continuous}, \ref{mean value},
and \ref{sigma-additivity} hold. The uniquely determined no-arbitrage price is given
by $\pi=\pi'$, where $P_t = f(t)$.
\end{proposition}

\begin{proof}
By \eqref{2}, $\pi'$ is linear on $\mathcal{M}_b$ and $\pi'(\gamma)>0$ for $\gamma\in\mathcal{M}^+_b\setminus\{o\}$. Furthermore,  $\pi'(\delta_0) = f(0) = 1$.
Theorem \ref{Theo4} therefore implies that a {\bf NA} market is defined by \eqref{42}, and 
that $\pi = \pi'$ and $P_t = f(t)$. 
The properties stated in the Assumptions \ref{P_t continuous}
and \ref{mean value} follow immediately.
For $k=1,2,\ldots$ let $\beta_k, \gamma_k\in\mathcal{M}_b^+$, let
$\beta_k \sim \gamma_k$, and let
$\beta=\sum_{k=1}^\infty \beta_k, \gamma=\sum_{k=1}^\infty \gamma_k\in\mathcal{M}_b^+$. 
It  follows with \eqref{4} that
\begin{eqnarray}
\pi(\beta) 
& = &
\int_0^\infty f\, d\beta
=
\sum_{k=1}^{\infty} \int_0^\infty f\, d\beta_k 
=
\sum_{k=1}^\infty \pi(\beta_k) 
<
\infty ,
\end{eqnarray}
and, similarly, 
$\pi(\gamma) 
=
\sum_{k=1}^\infty \pi(\gamma_k) 
<
\infty$,
where all $\pi(\beta_k)=\pi(\gamma_k) \geq 0$. Therefore, $\pi(\beta) = \pi(\gamma)$, 
and Assumption \ref{sigma-additivity} follows with \eqref{42}, since $\beta\sim\gamma$.
\end{proof}

With Proposition \ref{existence}, it is now clear, how an arbitrage-free market for deterministic cash flows, 
which exhibits the uniquely determined price functional \eqref{main}, can be constructed. For instance, one could
set up an environment with a constant effective annual rate of interest $i\geq 0$ by assuming
\begin{equation}
f(t) = (1+i)^{-t} ,\quad t \geq 0 ,
\end{equation}
where $t$ is measured in years. Under this term structure, the {\bf NA} price given by \eqref{main}
is consistent with the standard formulae for so-called present values of typical cash flows which can
be found in the actuarial literature.

For a very simple discrete example, consider the price of an annuity that pays annually one currency 
unit in arrears for $n$ years, i.e. the present value $a_{\yr{n}}$ of
\begin{equation}
\gamma = \sum_{k=1}^n \delta_k .
\end{equation}
Using international actuarial notation (with $v = 1/(1+i))$, \eqref{main} yields the standard result
\begin{equation}
a_{\yr{n}} 
= 
\pi(\gamma) = \sum_{k=1}^n P_k = \sum_{k=1}^n v^k = \frac{1-v^n}{i} ,
\end{equation}
(Gerber (1997), p.~9, McCutcheon and Scott (1986), p.~45). 
A time-continuously payable annuity of length $n$, where
$\gamma = \lambda_{[0,n]}$,
with $\lambda_{[0,n]}$ the trace Lebesgue measure on $[0,n]$, 
results in the well-known present value of
\begin{equation}
\label{cont_ann2}
\overline{a}_{\yr{n}} 
= 
\pi(\lambda_{[0,n]}) 
= 
\int_0^n P_t\, dt = \frac{1-v^n}{\log(1+i)} ,
\end{equation}
since $P_t = v^t = \exp(-t\cdot \log(1+i))$ (p.~51 in McCutcheon and Scott, 1986).
If a cash flow was given as a combination $\gamma = \gamma_1 + \gamma_2$, where 
\begin{equation}
\gamma_1 = \sum_{k=1}^n c_{t_k}\delta_{t_k} , \quad
c_{t_k}\in\mathbb{R} \text{ for all } k=1,\ldots,n ,
\end{equation}
and where $\gamma_2$ was given in terms of a signed Radon-Nikodym density $\rho$ w.r.t. the
Lebesgue measure $\lambda$, i.e.
\begin{equation}
\gamma_2(B) = \int_B \rho(t) d\lambda(t) \quad \text{for all } B\in\mathcal{B} ,
\end{equation}
where it can be assumed that $\rho$ has none, or only finitely many, discontinuities, then,
because of
\begin{equation}
\pi(\gamma_2) 
=
\int_0^\infty P_t d\gamma_2
=
\int_0^\infty P_t\rho(t) d\lambda(t)
=
\int_0^\infty \rho(t)P_t dt ,
\end{equation}
\eqref{main} turns into 
\begin{equation}
\pi(\gamma) 
=
\pi(\gamma_1) + \pi(\gamma_2) 
=
\sum_{k=1}^n c_{t_k}P_{t_k} 
+
\int_0^\infty \rho(t)P_t dt .
\end{equation}
Again, this formula can be found in textbooks, e.g.~McCutcheon and Scott (1986), p.~22, 
in the case of a flat term structure with the notation $v(t) = (1+i)^{-t}= P_t$.


\section{Counterexamples} 

\label{Counterexamples} 

This section describes an example of an arbitrage-free market, where price formula \eqref{main}
does not hold. This is mainly achieved by dropping Assumption \ref{mean value},
but first some technical preliminaries are needed.

Any measure $\mu\in\mathcal{M}_b^+(\R_0^+)$ has a unique so-called 
Lebesgue decomposition
\begin{equation}
\label{gammagammagamma}
\mu = \ol\mu + \hat\mu ,
\end{equation}
with $\ol\mu\in\mathcal{M}_b^+$ a measure continuous w.r.t.~$\lambda=\lambda_{\R_0^+}$, 
i.e.~there exists a Radon-Nikodym derivative, $\frac{d\ol\mu}{d\lambda}$,
and with $\hat\mu\in\mathcal{M}_b^+$ orthogonal to $\lambda$,
in the sense that there exists a Lebesgue nullset
$N\in\mathcal{B}$, i.e.~$\lambda(N)=0$, such that $\hat\mu(\R_0^+\setminus N)=0$.
For $\mu\in\mathcal{M}_b(\R_0^+)$ and $\mu = \mu^+ - \mu^-$ with 
$\mu^+,\mu^-\in\mathcal{M}_b^+$, one obtains
\begin{eqnarray}
\label{gammagammagammagammagamma}
\mu 
& = &
(\ol\mu^+ - \ol\mu^-) + (\hat\mu^+  - \hat\mu^-) \\
\nonumber
& = &
\ol\mu + \hat\mu ,
\end{eqnarray}
where $\ol\mu$ and $\hat\mu$ have the same properties as before, i.e.~there is a signed 
Radon-Nikodym derivative for $\ol\mu$ w.r.t.~$\lambda$, and there exists 
a Lebesgue nullset $N\in\mathcal{B}$ (as the union of the corresponding nullsets for 
$\mu^+$  and $\mu^-$) 
such that for any Borel set $B\in\mathcal{B}$ with $B\subset\R_0^+\setminus N$
one has $\hat\mu(B)=0$.
In the proposition below, it is shown that this decomposition exists for signed measures and
is unique as well. 
Note that the uniqueness in the signed case is possibly not entirely obvious, since a signed measure
can have many representations as a difference of two non-negative measures, of which
the Hahn-Jordan decomposition is only a special (minimal) case.

\begin{proposition}[Lebesgue decomposition for finite signed Borel measures]
\label{mumumu_2}
Let $\mu\in\mathcal{M}_b(\R_0^+)$, then there is a unique decomposition
\begin{equation}
\label{mumumu}
\mu = \ol\mu + \hat\mu 
\end{equation}
with $\ol\mu\in\mathcal{M}_b(\R_0^+)$ having a signed Radon-Nikodym derivative 
and $\hat\mu\in\mathcal{M}_b(\R_0^+)$ being orthogonal
w.r.t.~$\lambda_{\R_0^+}$ in the sense that there exists 
a Lebesgue nullset $N\in\mathcal{B}$ 
such that for any Borel set $B\in\mathcal{B}$ with $B\subset \R_0^+\setminus N$ 
one has $\hat\mu(B)=0$.
\end{proposition}

\begin{proof}
By contradiction. Existence is clear from the preceeding remarks.
Assume now that there are two such decompositions, 
$\mu = \ol\mu + \hat\mu = \ol\mu' + \hat\mu'$. Hence, 
\begin{equation}
\label{olmuolmuhatmuhatmu}
(\ol\mu - \ol\mu') = (\hat\mu' - \hat\mu) ,
\end{equation}
and there exists a Lebesgue nullset
$N\cup N'$ such that $(\hat\mu' - \hat\mu)(B)=0$ for any Borel set $B\subset\R_0^+\setminus (N\cup N')$.
Consider now any Borel set $A\subset\R_0^+$ and the disjoint union $A = \ol A \cup \hat A$ with 
$\ol A = A \cap (N\cup N')$, and $\hat A = A \cap (\R_0^+\setminus (N\cup N'))$. 
Using \eqref{olmuolmuhatmuhatmu},
\begin{eqnarray}
(\ol\mu - \ol\mu')(A) 
= 
(\ol\mu - \ol\mu')(\ol A \cup \hat A) 
& = & (\ol\mu - \ol\mu')(\ol A) + (\ol\mu - \ol\mu')(\hat A) \\
\nonumber
& = &
0 + (\hat\mu' - \hat\mu)(\hat A) = 0 .
\end{eqnarray}
Therefore, $\ol\mu = \ol\mu'$ and, by \eqref{olmuolmuhatmuhatmu}, $\hat\mu' = \hat\mu$.
\end{proof}

For a counterexample, let now $f$ and $g$ be bounded continuous functions
$\R_0^+ \rightarrow \R^+$ with $f(0)=1$, and generally $f \neq g$, for instance
one could choose $g = 2f$. Consider now the functional
\begin{eqnarray}
\label{pi1pi2}
\tilde\pi:  \mathcal{M}_b(\R_0^+) & \rightarrow & \R \\ \index{$\pi$}
\nonumber \gamma 
& \mapsto & 
\int_0^\infty g(t)\, d\ol\gamma(t) + \int_0^\infty f(t)\, d\hat\gamma(t) ,
\end{eqnarray}
where $\gamma = \ol\gamma + \hat\gamma$
is the Lebesgue decomposition of Proposition \ref{mumumu_2}.
By Eq.~\eqref{2},
\eqref{pi1pi2} is linear in $\gamma$ since a union of two Lebesgue nullsets is again a nullset,
and therefore, if $\gamma = a\alpha + b\beta$ for $\gamma, \alpha, \beta\in\mathcal{M}_b$
and $a,b\in\R$, then
$\ol\gamma = a\ol\alpha + b\ol\beta$
and 
$\hat\gamma = a\hat\alpha + b\hat\beta$.
Define now an equivalence relation
\begin{equation}
\gamma_1\sim\gamma_2 \; \Leftrightarrow \; \tilde\pi(\gamma_1) = \tilde\pi(\gamma_2) 
\end{equation}
for $\gamma_1, \gamma_2\in\mathcal{M}_b$, and note that 
$\ol\delta_t = o$, $\hat\delta_t = \delta_t$\index{$\hat\delta_t$} and
$\tilde\pi(\delta_0)=f(0)=1$.
By Theo.~\ref{Theo4}, {\bf NA}\, holds on the market $(\mathcal{M}_b, \sim)$, and 
$\pi = \pi_{\sim} = \tilde\pi$, as well as $P_t = P_{\sim,t} = f(t)$.
However, $f \neq g$ will quite obviously generally mean that
\begin{equation}
\pi(\gamma)  \neq  \int_0^\infty P_t\, d\gamma(t) .
\end{equation}
A simple example for this would be the continuously paid annuity $\gamma = \lambda_{[0,n]}$,
where the assumption $g(t) = 2f(t) = 2(1+i)^{-t}$, which violates Assumption \ref{mean value},
would now -- in contrast to \eqref{cont_ann2} -- result in a {\bf NA} price of
\begin{equation}
\pi(\lambda_{[0,n]}) = \int_0^n 2P_tdt = 2\frac{1-v^n}{\log(1+i)} = 2\overline{a}_{\yr{n}} .
\end{equation}
It should be noted that this example, of course, does not contradict Choquet's Theorem.
As pointed out at the end of Sec.~\ref{No-arbitrage pricing relative to zero-coupon bonds},
the requirements of Choquet's Theorem are not given in the here considered setup, which is the
reason that, despite of the generally valid Choquet representation \eqref{Choquet representation},
not every linear functional on $\mathcal{M}_b$ needs to be of the corresponding Choquet form,
and extra requirements, such as Assumption \ref{mean value}, are needed to obtain \eqref{main}.


\section{Forward rates and prices of forward cash flows} 

\label{Forward rates} 

Before dealing with FX cash flows or streams of commodities in the next section, 
note that in an arbitrage-free market $(\mathcal{M}_b, \sim)$ it holds for $\gamma\in\mathcal{M}_b$ that
for $t\geq 0$
\begin{equation}
P_t\gamma \sim P_t\pi(\gamma)\delta_0 \sim \pi(\gamma)\delta_t
\end{equation}
which implies
\begin{equation}
\gamma \sim \frac{\pi(\gamma)}{P_t}\delta_t ,
\end{equation}
and justifies, why $\pi(\gamma)/P_t$ can be called the time $t$ forward price of $\gamma$. 
However, it is important to note here that in this trade, $\gamma$ must be delivered (handed over)
immediately, while the amount $\pi(\gamma)/P_t$ is only paid in $t$. The price $\pi(\gamma)/P_t$
would only be a forward price in the actual sense, if the cash flow $\gamma$ was zero before time $t$, 
i.e.~if $\gamma_{[0,t)} = o$, since then $\gamma$ could only be delivered at time $t$ as well.
With this in mind, the following definition can be made.

\begin{definition}[Forward price of cash flows]
\label{def:forward price}
In an arbitrage-free market $(\mathcal{M}_b, \sim)$,
\begin{equation}
\label{forward price}
\pi_t = \frac{\pi}{P_t} , \quad t\geq 0 ,
\end{equation}
is called the time $t$ forward price.
\end{definition}

Obviously, the so far considered {\bf NA} price $\pi$ is identical to $\pi_0$, and
$\pi_t(\delta_s) = P_s/P_t$. By Corollary \ref{cor:Change of numeraire}, it is also
clear that the forward price $\pi_t$ simply is the {\bf NA} price under the numeraire
$\delta_t$, i.e.~$\pi_t = \pi_{\delta_t}$.

\begin{corollary}[Consistency of forward prices]
\label{Consistency of forward prices}
In an arbitrage-free market $(\mathcal{M}_b, \sim)$,
$\pi_t$ is linear for all $t\geq 0$, and $\pi_t(o)=0$. Furthermore,
\begin{eqnarray}
\label{forward price2}
\pi_s(\gamma)
& = &
\pi_s(\pi_t(\gamma)\delta_t)  , \quad s, t\geq 0 , \gamma\in\mathcal{M}_b .
\end{eqnarray}
\end{corollary}

\begin{proof}
\begin{equation}
\pi_s(\pi_t(\gamma)\delta_t) 
=
\pi_t(\gamma)\pi_s(\delta_t)  
=
\frac{\pi_0(\gamma)\pi_0(\delta_t)}{P_tP_s}
=
\frac{\pi_0(\gamma)}{P_s} .
\qedhere
\end{equation}
\end{proof}
Eq.~\eqref{forward price2} means that the time $s$ forward value of a cash flow
always equals the time $s$ forward value of the payment of that cash flow's time $t$ forward value in $t$. 
This obviously implies $\pi_s(\gamma - \pi_t(\gamma)\delta_t) = 0$, 
which is a generalization of \eqref{gamma-pi(gamma)delta0}. It also implies that
$\pi_0(\gamma) = \pi_0(\pi_s(\gamma)\delta_s) = \pi_0(\pi_t(\gamma)\delta_t)$ for $t>s>0$, meaning 
that no arbitrage can be achieved by selling forward and buying back forward a cash flow
at different times, since the present value of such forward deals is always the same as the present value
of the cash flow.

Under the conditions of Theorem \ref{Theo3},
one can for $t > s \geq 0$ also define the effective forward rates
\begin{equation}
\label{eq:forward rates}
f_{s,t} 
= 
\left(\frac{P_s}{P_t}\right)^{\frac{1}{t-s}} - 1 ,
\end{equation}
which deserve their name since, with $\delta_s/P_s \sim \delta_t/P_t \sim \delta_0$,
\begin{equation}
\label{762}
\delta_s 
\sim 
\delta_t(1 + f_{s,t})^{t-s} .
\end{equation}
Further, define $f_{t,t} = 0$. 
In the special case of $s=0$, the rate $y_t := f_{0,t} = (1/P_t)^{1/t}-1$ is called the effective spot rate
for the maturity $t$, since this rate is the yield of a zero-coupon bond maturing at $t$. 
From \eqref{762} one then obtains the well-known equality
\begin{equation}
\label{77}
(1 + y_t)^{t}
=
(1 + y_s)^{s}(1 + f_{s,t})^{t-s} ,
\end{equation}
which in turn implies for $t,r,s \geq 0$ that
\begin{equation}
(1 + f_{r,r+s+t})^{s+t}
=
(1 + f_{r,r+s})^{s}
(1 + f_{r+s,r+s+t})^{t} .
\end{equation}

Still under the conditions of Theorem \ref{Theo3}, Definition \ref{def:forward price} yields
\begin{eqnarray}
\label{79}
\pi_t(\gamma) 
=
\frac{\pi(\gamma)}{P_t}
& = &
\frac{1}{P_t} \int_0^\infty P_s\, d\gamma(s) \\
\nonumber
& = & 
\int_{[0,t)} (1+f_{s,t})^{t-s}\, d\gamma(s) 
+
\int_t^\infty (1+f_{t,s})^{t-s}\, d\gamma(s) .
\end{eqnarray}
If the cash flow $\gamma$ is zero before time $t$, then \eqref{79} turns into the simpler
\begin{equation}
\label{eq:92-2}
\pi_t(\gamma) 
= 
\int_0^\infty (1+f_{t,s})^{t-s}\, d\gamma(s) .
\end{equation}
Price formula \eqref{main} is simply the special case where $t=0$, 
i.e.~where $(1+f_{t,s})^{t-s} = (1+f_{0,s})^{-s} = (1+y_s)^{-s} = P_s$.

Finally, a remark on arbitrage opportunities characterized as ``beating the risk-free rate of return".
With the introduced notation, the risk-free growth factor achievable over a time period from $s$
to $t$, $t>s$, is $(1+f_{s,t})^{t-s}$. Investing at $s$ and getting all money out at $t$,
an investor in the here presented model would have to realize a cash flow of 
$r\delta_t - \delta_s$ with $r>(1+f_{s,t})^{t-s}$ at no cost
to beat the risk-free rate of return. However, by \eqref{eq:forward rates},
\begin{equation}
\pi_0(r\delta_t - \delta_s) 
= 
rP_t - P_s 
= 
rP_t - (1+f_{s,t})^{t-s}P_t
>
0 ,
\end{equation}
such that this kind of arbitrage is not possible while excluding \eqref{eq:arbitrage}.
Moreover, for a strategy where $\gamma\in\mathcal{M}_b^+([s,t])\setminus\{o\}$
was purchased for $\pi_r(\gamma)$ at time $r\in [0, s]$, a yield (internal rate of return) 
can be defined as being given by a constant $i$ such that 
\begin{equation}
\int_s^t (1+i)^{r-u}d\gamma(u) = \pi_r(\gamma) .
\end{equation}
With $f_\text{max} = \max_{s\leq u \leq t} f_{r,u}$, it holds by \eqref{eq:92-2} that
\begin{equation}
\pi_r(\gamma)
\geq
\int_s^t (1+f_\text{max})^{r-u}d\gamma(u) ,
\end{equation}
which implies $i \leq f_\text{max}$. Therefore, the investor's yield does not beat the 
highest of all (risk-free) forward rates $f_{r,u}$, $u\in [s,t]$.


\section{Norberg's theory of cash flow valuation} 

\label{Norberg's theory of consistent cash flow valuation} 

In 1990, Ragnar Norberg published an article in the Scandinavian Actuarial Journal in which he presented a theory
of ``consistent" valuation functions for cash flows (Norberg, 1990). Since Norberg's resulting valuation principle is
for certain cash flows, and when interpreted the right way, identical to the one presented here, this section will
explore the parallels in more detail. Before doing so, it is pointed out that only the first half of Norberg's paper
is referred to here, while the second half, in which he considers inter-relationships of certain cash flows and
probability distributions on payment measures, goes beyond the theory presented in here.

First, it should be mentioned that Norberg's perspective was a more actuarial one. His approach does not
present itself as a no-arbitrage theory in the modern sense, although the requirements posed on valuation principles
for being ``consistent" will turn out to be analogs to some of the no-arbitrage results presented earlier. 
Secondly, Norberg's intention in the first part of his paper is to derive a consistent theory of the deterministic
valuation of deterministic cash flows at any point on the real time axis. As such, what in this paper was
presented as forward prices, he presents as deterministic prices in the future, but, of course, in reality the value of a 
cash flow in the future is unknown, since the future term structure and the future value of received
payments must be unknown, for the latter ones, because it is unknown how they will be reinvested up to the point of 
valuation.
Similarly, he considers a past (``previous investments can be cashed with addition of earned interest"; Norberg, 1990), 
while this article had to neglect the past for the just explained reason of obviously unknown
reinvestment strategies.
As such, there is a philosophical difference between the here presented theory and Norberg (1990),
which should be kept in mind.

The cash flows Norberg (1990) considers are modeled as $\sigma$-finite, non-negative Borel measures on $\mathbb{R}$,
which have a non-decreasing payment function ``$B$" in the sense that any time $t$ can be mapped to the finite measure
lying on $(-\infty, t]$ (compare \eqref{muFF3}).
His cash flows are therefore less general in the sense that they are only positive, but more general in that the
total amount paid can be infinite.
To any such cash flow he then assigns a value $V(t,B)$, ``which is the single payment against which $B$ can be
exchanged at time $t$" (Norberg, 1990). So, the first parallel to draw would be that the earlier introduced no-arbitrage
forward price $\pi_t(\cdot)$ may be some analog to Norberg's $V(t,\cdot)$ as long as 
cash flows in $\mathcal{M}_b^+$ are considered, and $t \geq 0$. However, there are several more
analogs, which the reader can find in Table \ref{tab:Comparison of notation}.

\begin{table}[htb]
\begin{center}
\begin{tabular}{c|c||c|c}
\multicolumn{2}{c||}{Norberg (1990)} & \multicolumn{2}{c}{This article} \\
\hline
Description & \multicolumn{2}{c|}{Symbol} & Description \\
\hline\hline
payment function & $B$ & $F_{\gamma}$ & distribution function \\
value & $V(t,B)$ & $\pi_t(\gamma)$, $t\geq 0$ & {\bf NA} forward price \\
& $V(0,B)$ & $\pi(\gamma)=\pi_0(\gamma)$ & {\bf NA} price \\
& & $o$ & null measure \\
null stream & $B_0$ & $F_o$ & \\
unit mass concentr.~at $t$ & $\varepsilon_t$ & $\delta_t$ & Dirac measure in $t$\\
$(t-u)$-y.~int.~fact.~at $t$ & $v(t,u) = V(t,\varepsilon_u)$ & $\pi_t(\delta_u) = P_u/P_t$\\
discount function & $v(t)=v(0,t)$ & $P_t$ & UZCB price
\end{tabular}
\end{center}
\caption{\label{tab:Comparison of notation}Comparison of notation. Symbols on the right are only defined on the non-negative
time axis.}
\end{table}

Norberg now calls the ``valuation function" ``consistent" if five properties are fulfilled:

{[i]} $V(t,B_0) = 0$, which, by Table \ref{tab:Comparison of notation}, corresponds to 
$\pi_t(o)=0$, $t\geq 0$, which clearly holds for the  {\bf NA} forward price \eqref{forward price}.

{[ii]} $V(t,B_1) \leq V(t,B_2)$ if the measure belonging to $B_2$ is not smaller than that belonging to $B_1$,
with strict inequality if the total amount paid by $B_2$ is larger than that of $B_1$. For $t\geq 0$,
this is implied by the last statement of Corollary \ref{cor:1} of this paper.

{[iii]} $\sigma$-additivity of $V(t,B)$ in $B$ if a countable sum of payment measures is again a 
payment measure in Norberg's sense.
This is Lemma \ref{lemma:sigma-additivity}, which was implied by Assumption \ref{sigma-additivity}.

{[iv]} A finite value for measures of finite support. This obviously holds in the presented model since
all prices are finite.

{[v]} $V(s, B) = V(s, V(t,B)\varepsilon_t)$, which has its analog in \eqref{forward price2} of  
Corollary \ref{Consistency of forward prices}.

Therefore, all the properties that define Norberg's consistency are fulfilled by the no-arbitrage price of
the here presented model, under the obvious restriction that not the mathematically exact same 
cash flows are considered. An important difference to Norberg (1990) is, that these properties were 
not {\em a priori} required, but they follow as a consequence of the basic market definition, the no-arbitrage
assumption, and the additional Assumption \ref{sigma-additivity}

Norberg then goes on by defining a ``consistent" discount function $v$ as a strictly positive function, which
is bounded on finite intervals, and for which $v(0) = 1$. In his Theorem 1, he then shows that
consistent discount functions define consistent valuation functions through
\begin{equation}
\label{Norberg1}
V(t,B) = \frac{1}{v(t)}V(B)
\end{equation}
with $V(B) = V(0,B)$ and
\begin{equation}
\label{Norberg2}
V(B) = \int v dB .
\end{equation}
This result is comparable to Proposition \ref{existence}, which showed how the Choquet formula \eqref{main}
or \eqref{Norberg2} can define an arbitrage-free market via the corresponding arbitrage-free price, where 
the integrand $f$ had to fulfill certain conditions (more than Norberg required for $v$, more details follow below
in the ``regular" case).
Furthermore, his Theorem 1 shows that consistent valuation functions determine consistent discount functions through 
$v(t,u) = V(t,\varepsilon_u)$ and $v(t)=v(0,t)$. An analog property follows in the here presented model
for the UZCB prices under no-arbitrage, which are strictly positive and equal $1$ for the immediately paying bond,
however, without further assumptions, no boundedness on finite intervals as in Norberg's case would follow.
Moreover, it is shown that if $V$ is a consistent valuation function, then
\eqref{Norberg1} and \eqref{Norberg2} hold for discrete payment measures.
Obviously, \eqref{Norberg1} corresponds to \eqref{forward price}, and \eqref{Norberg2} to \eqref{main},
which hold for the no-arbitrage price of any cash flow (not just discrete ones) under the assumptions of 
Theorem \ref{Theo3}, which are of course stricter than Norberg's.

He then carries on to define a ``regular" discount function as a consistent one which is continuous and non-increasing.
This corresponds to Assumption \ref{P_t continuous}, which in fact is weaker, since only boundedness
(and not monotony) is required. ``Regular" valuation functions he then defines via two additional
properties:

{[vi]} $V(t,B)$ should be continuous in $t$ for any $B$. Since Norberg has before shown that \eqref{Norberg1}
holds in general, this implies continuity of $v$, and vice versa.

{[ii']} $V(t,B_1) \leq V(t,B_2)$ if $B_2(t) \geq B_1(t)$, $t\in\mathbb{R}$,
for the payment function, and with strict inequality if the
total amount paid by $B_2$ is larger than that of $B_1$. This property -- Norberg refers to the
saying ``time is money" -- values earlier payments higher than latter ones. In the here presented theory,
this was not required, and also does not follow, for the simple reason that the model allows for negative interest
rates and thus for UZCB prices that are larger than $1$. In the middle of the second decade of this millenium,
negative interest rates have become an economic reality that cannot be ignored.

In Theorem 2, Norberg then proves that for any regular discount function the valuation function defined
via \eqref{Norberg1} and \eqref{Norberg2} is regular. This is not surprising, given that 
non-increasing continuous $v$ was required. A comparable statement is Proposition \ref{existence} in the here
presented theory, which lacks the ``time is money" property for the earlier explained reason.
Furthermore, it is shown that  the discount function of a regular valuation function
is regular, which again has no direct analog in this paper because of the possibility of negative interest rates.
However, if in analogy one required $\pi_t(\gamma)$ to be continuous in $t$ for any $\gamma$, then it would
similarly follow from \eqref{forward price} that $P_t$, the analog of $v$, had to be continuous.
Moreover, Norberg's Theorem 2 then shows that, for a regular valuation function, $V(t,B)$ is for any payment
measure uniquely determined by \eqref{Norberg1} and \eqref{Norberg2}, which obviously
corresponds to Theorem \ref{Theo3}. In this context it should be mentioned that continuity 
of the (forward) {\bf NA} price in $t$ followed in the here presented case from 
Assumption \ref{P_t continuous}, but to actually obtain Theorem \ref{Theo3} the
``average value requirement" of Assumption \ref{mean value} was needed as well, while
Norberg needed the additional requirement {[ii']}. So, even where
direct comparison must fail (for allowing negative interest rates), there still exist parallels
between the two approaches.

Summing up, in addition to the identical valuation formula of Choquet type, there are some striking similarities to
Norberg's paper, which the author of this work
did not know at the time when most of the here presented results were obtained. However, the article was
improved after reading Norberg's inspiring piece, such that the parallels are more obvious now.
The assumptions made in this article, namely Definition \ref{market} for the basic market, 
no-arbitrage, and the additional Assumption \ref{sigma-additivity}, implied Norberg's ``consistency", and 
-- when Assumption \ref{P_t continuous} is added -- in parts ``regularity" for the 
valuation principle and the UZCB prices (discount function). The presented approach can possibly be seen
as a modern version of Norberg's, with assumptions that seem somewhat more basic, or natural,
given the background of arbitrage theory. 


\section{Forward markets and combined markets} 

\label{Forward markets} 

The main result of \eqref{main} was derived without the specification of any currency unit, which
means that it holds for any market of deterministic `asset flows'. For instance, one could consider a market
in which deterministic streams of natural gas or crude oil are exchanged for one another. 
Similarly, different currencies
could be considered, such as U.S.--Dollars and Euros. For instance, a European money market 
$(\mathcal{M}_b, \stackrel{\text{\euro}}{\sim})$ and a U.S.~money market 
$(\mathcal{M}_b, \stackrel{\$}{\sim})$
could be considered as a combined, or common, market with the obvious linear operations on
cash flows $(\gamma_\text{\euro}, \gamma_\$)\in\mathcal{M}_b \oplus \mathcal{M}_b$, 
which represent the simultaneous
cash flow of $\gamma_\text{\euro}$ in Euros and $\gamma_\$$ in U.S.--Dollars.

\begin{definition}[Combined markets]
\label{def:combined markets}
\begin{enumerate}
\item[]
\item
$(\mathcal{M}_b \oplus \mathcal{M}_b, \sim)$ is called a combined market
if $\sim$ has the Properties 1--4 of Definition \ref{market}, where Property 4
is understood such that for any 
$(\gamma_1, \gamma_2)\in \mathcal{M}^+_b \oplus \mathcal{M}^+_b \setminus\{(o,o)\}$
there exists $b > 0$ such that $(\gamma_1, \gamma_2) \sim b(\delta_0, o)$.
\item
For $(\gamma_1, \gamma_2)\in \mathcal{M}_b \oplus \mathcal{M}_b$, $b\in\mathbb{R}$ is called
a 1-price of $(\gamma_1, \gamma_2)$
if $(\gamma_1, \gamma_2) \sim b(\delta_0, o)$, and a 2-price
if $(\gamma_1, \gamma_2) \sim b(o, \delta_0)$.
\item
There exists arbitrage in the combined market
$(\mathcal{M}_b \oplus \mathcal{M}_b, \sim)$ if
\begin{equation}
(o,o)\sim(\gamma_1, \gamma_2) \quad \text{ for at least one }
(\gamma_1, \gamma_2)
\in \mathcal{M}^+_b \oplus \mathcal{M}^+_b \setminus\{(o,o)\} .
\end{equation}
\end{enumerate}
\end{definition}

\begin{proposition}[Results for combined markets]
\label{lemma:forward_market}
Lemma \ref{lemma:inversion} and \ref{lemma:equivalence}, Proposition \ref{price existence}, 
Theorem \ref{LOP_theo}, \ref{theo:linearity}, and \ref{theo:pi_linear},
and Corollary \ref{cor:1}, \ref{cor:Change of numeraire}, and \ref{Consistency of forward prices}
apply to the combined market ${(\mathcal{M}_b \oplus \mathcal{M}_b, \sim)}$
regarding 1-prices, if $(o, o)$ takes the role of $o$,
$\mathcal{M}^+_b \oplus \mathcal{M}^+_b \setminus\{(o,o)\}$ takes the role of
$\mathcal{M}^+_b \setminus\{o\}$, $(\delta_0, o)$ takes the role of $\delta_0$, and
if $P^1_t = \pi_1(\delta_t, o)$ takes the role of $P_t$.
\end{proposition}

\begin{proof}
$\mathcal{M}_b \oplus \mathcal{M}_b$ is a real vector space with the null vector $(o,o)$.
The positive cone $\mathcal{M}^+_b \oplus \mathcal{M}^+_b\subset\mathcal{M}_b \oplus \mathcal{M}_b$
fulfills Assumption \ref{assu:cone}, because any $(\gamma_1, \gamma_2)\in \mathcal{M}_b \oplus \mathcal{M}_b$
has a Hahn-Jordan type decomposition $(\gamma_1, \gamma_2) = 
(\gamma^+_1, \gamma^+_2) - (\gamma^-_1, \gamma^-_2)$, with
$(\gamma^+_1, \gamma^+_2), (\gamma^-_1, \gamma^-_2)\in
\mathcal{M}^+_b \oplus \mathcal{M}^+_b$. 
Moreover, $(o,o)\neq (\delta_0, o)\in\mathcal{M}^+_b \oplus \mathcal{M}^+_b$, which means
that the combined market of Definition \ref{def:combined markets} together $\gamma_0=(\delta_0, o)$
fulfills Assumption \ref{assu:relation} and \ref{assu:gamma-value}.
Thus, Lemma \ref{lemma:inversionc}, \ref{lemma:equivalencec}, and \ref{price existencec}, 
Proposition \ref{LOP_theoc} and \ref{theo:linearityc}, Theorem \ref{pi_linearc}, and Corollary \ref{cor:1c}
apply, and the listed results emerge in analogy.
\end{proof}

In analogy to Definition \ref{def:forward price}, a forward 1-price can be defined.

\begin{definition}[Forward 1-price of cash flows]
\label{def:forward i-price}
In an arbitrage-free market $(\mathcal{M}_b \oplus \mathcal{M}_b, \sim)$, where
$P^1_t = \pi_1(\delta_t, o)$,
\begin{equation}
\label{forward 1-price}
\pi_{1,t} = \frac{\pi_1}{P^1_t} , \quad t \geq 0 ,
\end{equation}
is called the time $t$ forward 1-price.
\end{definition}

\begin{corollary}[Existence of the 2-price under {\bf NA}]
\label{Existence of a 2-price under NA}
For an arbitrage-free combined market $(\mathcal{M}_b \oplus \mathcal{M}_b, \sim)$,
there exists a uniquely determined, linear (forward) 2-price which is strictly positive on 
$\mathcal{M}^+_b \oplus \mathcal{M}^+_b \setminus\{(o,o)\}$, and it holds that
\begin{equation}
\label{eq:Existence of the 2-price under NA}
\pi_{2,t} 
=
\frac{\pi_{1,t}}{\pi_{1,t}(o, \delta_t)} , \quad t \geq 0 .
\end{equation}
\end{corollary}

\begin{proof}
By Theorem \ref{pi_linearc}, since $\pi_1$ is a {\bf NA} price for the numeraire $(\delta_0, o)$,
$\pi_{1,t}$ is therefore a {\bf NA} price for the numeraire $(\delta_t, o)$, and
$\pi_{2,t}$ is consequently a {\bf NA} price for the numeraire $(o, \delta_t)$.
\end{proof}

For convenience of notation, the rest of this section is carried out with respect to a FX
forward market for the currency cross EUR/USD. However, all results can be applied
to any suitable market combination.

\begin{proposition}[{\bf NA} price in combined markets]
\label{prop:forward_market}
If the markets $(\mathcal{M}_b, \stackrel{\text{\euro}}{\sim})$ and 
$(\mathcal{M}_b, \stackrel{\$}{\sim})$ are free of arbitrage with prices
$\pi_\text{\euro}$ and $\pi_\$$ respectively, then there exists only one
arbitrage-free combined market $(\mathcal{M}_b \oplus \mathcal{M}_b, \sim)$
which fulfills for given $\text{EURUSD}_0 > 0$ that
\begin{eqnarray}
\label{64}
(\gamma_1, o)  \sim (\gamma_2, o)
& \Leftrightarrow &
\gamma_1 \stackrel{\text{\euro}}{\sim} \gamma_2 , \\
\label{65}
(o, \gamma_1)  \sim (o, \gamma_2)
& \Leftrightarrow & 
\gamma_1 \stackrel{\$}{\sim} \gamma_2 ,
\end{eqnarray}
and
\begin{equation}
\label{71}
(\delta_0, o) \sim \text{EURUSD}_0(o, \delta_0) .
\end{equation}
With $\text{USDEUR}_0 = 1/\text{EURUSD}_0$,
the uniquely determined \euro-price of a two-currency cash flow 
$(\gamma_\text{\euro}, \gamma_\$)\in\mathcal{M}_b \oplus \mathcal{M}_b$ is then given by
\begin{equation}
\label{76}
\pi_\text{\euro}(\gamma_\text{\euro}, \gamma_\$)
=
\pi_\text{\euro}(\gamma_\text{\euro})+\text{USDEUR}_0\pi_\$(\gamma_\$) ,
\end{equation}
and the also unique \$-price is given by
\begin{equation}
\label{96}
\pi_\$(\gamma_\text{\euro}, \gamma_\$)
=
\text{EURUSD}_0\pi_\text{\euro}(\gamma_\text{\euro}) 
+
\pi_\$(\gamma_\$) .
\end{equation}
\end{proposition}

\begin{proof}
Existence: Let $\text{USDEUR}_0 > 0$ and define for any $(\beta_\text{\euro}, \beta_\$) , (\gamma_\text{\euro}, \gamma_\$)
\in \mathcal{M}_b \oplus \mathcal{M}_b$ that
\begin{equation}
\label{77-2}
(\beta_\text{\euro}, \beta_\$) \sim (\gamma_\text{\euro}, \gamma_\$) 
\quad\Leftrightarrow\quad
\pi_\text{\euro}(\beta_\text{\euro}, \beta_\$) = \pi_\text{\euro}(\gamma_\text{\euro}, \gamma_\$) ,
\end{equation}
where $\pi_\text{\euro}(\cdot, \cdot)$ as in \eqref{76}. Note that \eqref{76}
is linear on $\mathcal{M}_b \oplus \mathcal{M}_b$, and $\pi_\text{\euro}(\delta_0, o) = 1$.
Since
\begin{equation}
\pi_\text{\euro}(\gamma_\text{\euro}, \gamma_\$)
>
0
\quad\text{for all } (\gamma_\text{\euro}, \gamma_\$)
\in \mathcal{M}^+_b \oplus \mathcal{M}^+_b \setminus\{(o,o)\}
\end{equation}
by inheritance from $\pi_\text{\euro}(\gamma_\text{\euro})$ and $\pi_\$(\gamma_\$)$,
Theorem \ref{theo:pi_linear} can now be applied with respect to the 1-price $\pi_\text{\euro}$
(see Proposition \ref{lemma:forward_market}).
The properties \eqref{64}, \eqref{65}, and \eqref{71} follow directly from \eqref{77-2}.\\
Uniqueness:
If there exists an arbitrage-free combined market, then
(by Theorem \ref{theo:pi_linear} via Proposition \ref{lemma:forward_market})
there exists a unique {\bf NA} \euro-price,
and the properties \eqref{64}, \eqref{65}, and \eqref{71} imply
for any $(\gamma_\text{\euro}, \gamma_\$)\in\mathcal{M}_b \oplus \mathcal{M}_b$ that
\begin{eqnarray}
(\gamma_\text{\euro}, \gamma_\$)  
& \sim & 
(\gamma_\text{\euro}, o)  
+
(o, \gamma_\$)  \\
\nonumber & \sim & 
\pi_\text{\euro}(\gamma_\text{\euro})(\delta_0, o) +\pi_\$(\gamma_\$)(o, \delta_0) \\
\nonumber & \sim & 
(\pi_\text{\euro}(\gamma_\text{\euro})+\pi_\$(\gamma_\$)\text{USDEUR}_0)(\delta_0, o) .
\end{eqnarray}
Therefore, \eqref{76} is the only {\bf NA} 1-price, and \eqref{simm_2} of Theorem \ref{theo:pi_linear}
(via Proposition \ref{lemma:forward_market}) warrants uniqueness of the equivalence relation.
Eq.~\eqref{96} follows by $\pi_{\text{\euro},0}(o,\delta_0) = \text{USDEUR}_0$
from \eqref{eq:Existence of the 2-price under NA} in Corollary \ref{Existence of a 2-price under NA}.
\end{proof}

By \eqref{76} it is now clear that the default no-arbitrage method to price a FX (\$) cash flow $\gamma_\$$
in local currency (\euro) is to first discount payments in the foreign currency, and then exchange them at
the FX spot rate. Under the assumptions of Theorem \ref{Theo3} for both separate markets,
foreign (\$) and local (\euro), this can be written as
\begin{equation}
\label{712}
\pi_\text{\euro}(o, \gamma_\$)
=
\text{USDEUR}_0 \cdot\pi_\$(\gamma_\$)
= 
\text{USDEUR}_0\int_0^\infty P^\$_t\, d\gamma_\$(t) ,
\end{equation}
where $P^\$_t$ stands for the price of one U.S.--Dollar delivered at time $t$
($P^\text{\euro}_t$ is similarly defined for Euros). 

In analogy to \eqref{71}, define now for $t>0$ the under {\bf NA} uniquely determined 
ratio $\text{EURUSD}_t$ by
\begin{equation}
\label{711}
(\delta_t, o) \sim \text{EURUSD}_t(o, \delta_t) .
\end{equation}
Quite obviously, $\text{EURUSD}_t$ is the forward exchange rate of the FX cross EUR/USD,
meaning it is the forward price of EUR 1.00 in Dollars at time $t$. 
Using Definition \ref{def:forward price}, this could also be expressed with the forward price
$\pi_{\$,t}$, i.e.~$\text{EURUSD}_t=\pi_{\$,t}(\delta_t, o)$.
As usual, one defines $\text{USDEUR}_t = 1/\text{EURUSD}_t$. Then,
\begin{eqnarray}
P^\text{\euro}_t(\delta_0, o) 
\sim 
(\delta_t, o) 
& \sim &
\text{EURUSD}_t\cdot (o, \delta_t) \\
\nonumber & \sim &
\text{EURUSD}_t\cdot P^\$_t\cdot (o, \delta_0) \\
\nonumber & \sim &
\text{EURUSD}_t\cdot P^\$_t\cdot \text{USDEUR}_0\cdot (\delta_0, o) ,
\end{eqnarray}
which implies the well-known interest rate parity (e.g.~p.~113 in Hull, 2008)
\begin{eqnarray}
\label{752}
\text{USDEUR}_0 \cdot P^\$_t
& = &
\text{USDEUR}_t\cdot  P^\text{\euro}_t \\
\nonumber
\text{USDEUR}_0 \cdot (1+y_t^\text{\euro})^t
& = &
\text{USDEUR}_t\cdot  (1+y_t^\$)^t ,
\end{eqnarray}
where $y_t^\text{\euro}$ and $y_t^\$$ denote the spot rates in the corresponding 
currencies.

As an alternative to \eqref{712}, the following corollary explains how to price a FX cash flow by first converting it
by means of forward prices into (a hypothetical) cash flow in local currency, which is then discounted in the local
term structure.
The equivalence of the two pricing approaches, for instance in the case of time-discrete
cash flows, is of course well known (e.g.~pp.~167--169 in Hull, 2008). For the proof, another
technical assumption is necessary.

\begin{assumption}[Bounded forward rates]
\label{assu:4}
$\text{USDEUR}_t$ is bounded  on $ \R^+_0 $.
\end{assumption}

Note that under Assumption \ref{assu:4}, $\text{USDEUR}_t = \text{USDEUR}_0 \cdot P^\$_t/P^\text{\euro}_t$
is continuous and measurable in $t$, since $P^\$_t$ and $P^\text{\euro}_t$ are.

\begin{corollary}[Change of cash flow measure]
Consider the uniquely determined combined market with the properties 
\eqref{64}, \eqref{65}, and \eqref{71} of Proposition \ref{prop:forward_market},
and additionally suppose that -- besides {\bf NA} -- the Assumptions \ref{P_t continuous}, 
\ref{mean value}, and \ref{sigma-additivity} hold for the sub-markets 
$(\mathcal{M}_b, \stackrel{\text{\euro}}{\sim})$ and $(\mathcal{M}_b, \stackrel{\$}{\sim})$.
Furthermore, suppose that Assumption \ref{assu:4} holds. Then, 
\begin{equation}
\pi_\text{\euro}(o, \gamma_\$)
\; = \;
\text{USDEUR}_0 \cdot\pi_\$(\gamma_\$)
\; = \;
\pi_\text{\euro}(\gamma_\text{\euro})
\; = \;
\int_0^\infty P^\text{\euro}_t\, d\gamma_\text{\euro}(t) ,
\end{equation}
where $\gamma_\text{\euro}$ is the measure with the Radon-Nikodym derivative
$\text{USDEUR}_t$ w.r.t.~to $\gamma_\$$, i.e.
\begin{equation}
\label{RN}
\frac{d\gamma_\text{\euro}}{d\gamma_\$}(t) = \text{USDEUR}_t 
\quad\text{and}\quad
\gamma_\text{\euro}(B) = \int_B \text{USDEUR}_t \, d\gamma_\$ \text{ for all } B\in\mathcal{B} .
\end{equation}
\end{corollary}

\begin{proof}
By \eqref{712}, \eqref{752}, and \eqref{RN},
\begin{eqnarray}
\label{713}
\pi_\text{\euro}(o, \gamma_\$)
& = & 
\text{USDEUR}_0\int_0^\infty P^\$_t\, d\gamma_\$(t) \\
\nonumber
& = &
\int_0^\infty P^\text{\euro}_t\cdot\text{USDEUR}_t\, d\gamma_\$(t) \\
\nonumber
& = &
\int_0^\infty P^\text{\euro}_t\, d\gamma_\text{\euro}(t) 
\quad = \quad
\pi_\text{\euro}(\gamma_\text{\euro}) ,
\end{eqnarray}
where the existence of the last integral is given by Assumption \ref{assu:4}, as it implies
that $\gamma_\text{\euro}\in\mathcal{M}_b$ follows from
$\gamma_\text{\$}\in\mathcal{M}_b$.
\end{proof}


\section{Conclusion} 

This article has provided a mathematically rigorous no-arbitrage derivation of the price principles that
are commonly used in money markets and forward markets that deal with deterministic cash flows
or deterministic flows of other assets or commodities.
The generality of the here presented approach usually cannot be found in the literature,
however, a notable exception with certain parallels was discussed with Norberg (1990), even if 
this work was never intended as a no-arbitrage theory. Besides the presentation of comparatively few
necessary assumptions for the commonly used price formulae of Choquet-type,
this paper also gave sufficient conditions under which arbitrage-free market models, in which said price
formulae hold, indeed exist. Furthermore, it was shown that arbitrage-free models exist where the
generally accepted price formulae do not apply. A topic, which could not be considered, were cash flows
with infinite amounts, such as British perpetuities, and under what conditions those could be incorporated into this theory.



\end{document}